\documentclass{article}
\usepackage{amssymb}
\usepackage{amsfonts}
\usepackage{amsmath}

\newtheorem{theorem}{Theorem}

\newtheorem{proposition}[theorem]{Proposition}

\newenvironment{proof}[1][Proof]{\noindent\textbf{#1.} }{\ \rule{0.5em}{0.5em}}
\input{tcilatex}
\begin{document}

\title{Duality Mappings and Metric Extensor}
\author{A. Manuel Moya\thanks{%
e-mail: anmanumoya@gmail.com} \\
%EndAName
National Technological University-Haedo Regional Faculty\\
Paris 532, 1706 Haedo, Buenos Aires, Argentina}
\maketitle

\begin{abstract}
We introduce the key concepts of duality mappings and metric extensor. The
fundamental identities involving the duality mappings are presented, and we
disclose the logical equivalence between the so-called metric tensor and the
metric extensor. By making use of the duality mappings and the metric
extensor, we construct the so-called metric products, i.e., scalar product
and contracted products of both multivectors and multiforms. The so-known
identities involving the metric products are obtained. We find the
fundamental formulas involving the metric extensor and, specially, we try
its surprising inversion formula. This proposal unveils, once and for all,
an unsuspected meaning of the metric products.

\pagebreak
\end{abstract}

\tableofcontents

\pagebreak

\section{Introduction}

For the so-called vector spaces of multivectors and multiforms over a finite
dimensional real vector space, we introduce the key concept of duality
mappings: the pairing and the contractions, in sections 3 and 4. They all
apply pairs either (multiform,multivector) or (multivector,multiform) into
scalars, multivectors or multiforms. In section 5, the fundamental
identities involving the duality mappings are presented. We also introduce
another key concept: the metric extensor over a finite dimensional real
vector space, in section 6. And, we there disclose the logical equivalence
between the so-called metric tensor and the metric extensor. In section 6,
we also study in detail the so-called extension of the metric extensor,
because of it plays a fundamental roll for defining the so-called metric
products, i.e., scalar and contracted products of both multivectors and
multiforms. Then, by making use of the duality mappings and the metric
extensor, we construct those metric products, in sections 7, 8 and 9. In
section 10, the so-known identities involving the metric products are
obtained. Finally, in section 11, we find the fundamental formulas involving
the metric extensor and, specially, a surprising inversion formula for it.
In Appendix, we make the proofs of some identities and formulas for showing
"the tricks of the trade".

\section{Multiforms and multivectors}

Let $V$ be a real vector space of finite dimension, i.e., $\dim V=n\in 
\mathbb{N}.$ As it is usual, $V^{\star }$ will denote the dual vector space
of $V.$ Recall that $\dim V^{\star }=\dim V.$

Take a non-negative integer number $p,$ $\left( 0\leq p\leq n\right) .$ A $0$%
-form (or, scalar) over $V$ is a real number, the set of $0$-forms is
denoted by $\bigwedge^{0}V^{\star },$ i.e., $\bigwedge^{0}V^{\star }=\mathbb{%
R}.$

A $1$-form (or, form) over $V$ is a vector belonging to $V^{\star },$ the
set of $1$-forms is denoted by $\bigwedge^{1}V^{\star },$ i.e., $%
\bigwedge^{1}V^{\star }=V^{\star }.$ A $p$-form $\left( p\geq 2\right) $
over $V$ is a skew-symmetric $p$-covariant tensor over $V,$ the set of these
objects is usually denoted by $\bigwedge^{p}V^{\star }.$ As we know, each $%
\bigwedge^{p}V^{\star }$ is a real vector space of finite dimension, $\dim
\bigwedge^{p}V^{\star }=\dbinom{n}{p}.$ Sometimes, $2$-forms and $n$-forms
are respectively called biforms and pseudoscalars.

The direct product $\bigwedge^{0}V^{\star }\times \bigwedge^{1}V^{\star
}\times \cdots \times \bigwedge^{n}V^{\star },$ usually denoted by $%
\bigwedge V^{\star },$ is naturally a real vector space of finite dimension, 
$\dim \bigwedge V^{\star }=2^{n}.$

We call multiform over $V$ any vector belonging to $\bigwedge V^{\star }.$
It means that a multiform $\phi $ is exactly a $\left( n+1\right) $-upla $%
\left( \phi ^{0},\phi ^{1},\phi ^{2},\ldots ,\phi ^{n}\right) ,$ where $\phi
^{0}$ is a scalar, $\phi ^{1}$ is a form, $\phi ^{2}$ is a biform,$\ldots $
and $\phi ^{n}$ is a pseudoscalar.

Analogously, we define $0$-vector over $V,$ $1$-vector over $V$ and $p$%
-vector $\left( p\geq 2\right) $ over $V$ to be respectively a real number,
a vector belonging to $V$ and a skew-symmetric $p$-contravariant tensor over 
$V.$ The set of $p$-vectors is usually denoted by $\bigwedge^{p}V$ and, as
we know, is a real vector space of finite dimension too, $\dim
\bigwedge^{p}V=\dbinom{n}{p}.$ Sometimes, $0$-vectors, $2$-vectors and $n$%
-vectors are respectively called scalars, bivectors and pseudoscalars.

The direct product $\bigwedge^{0}V\times \bigwedge^{1}V\times \cdots \times
\bigwedge^{n}V,$ usually denoted by $\bigwedge V,$ is a real vector space of
finite dimension too, $\dim \bigwedge V=2^{n}.$

The vectors belonging to $\bigwedge V$ are called multivectors over $V.$
Therefore, a multivector $x$ is exactly a $\left( n+1\right) $-upla $\left(
x_{0},x_{1},x_{2},\ldots ,x_{n}\right) ,$ where $x_{0}$ is a scalar, $x_{1}$
is a vector, $x_{2}$ is a bivector,$\ldots $ and $x_{n}$ is a pseudoscalar.

We define the so-called $p$-component operators for multiforms and
multivectors.

For any $\phi \in \bigwedge V^{\star }$ and $x\in \bigwedge V:$%
\begin{equation*}
\text{if }\phi =\left( \phi ^{0},\phi ^{1},\phi ^{2},\ldots ,\phi
^{n}\right) ,\text{ then }\left[ \phi \right] ^{p}=\phi ^{p}
\end{equation*}%
\begin{equation}
\text{if }x=\left( x_{0},x_{1},x_{2},\ldots ,x_{n}\right) ,\text{ then }%
\left[ x\right] _{p}=x_{p}  \label{kcompmulti}
\end{equation}

Notice that $\left[ \left. {}\right. \right] ^{p}$ is a linear mapping from $%
\bigwedge V^{\star }$ to $\bigwedge^{p}V^{\star }$ and $\left[ \left.
{}\right. \right] _{p}$ is a linear mapping from $\bigwedge V$ to $%
\bigwedge^{p}V.$

We define the so-called $p$-homogeneous multiforms and multivectors.

Let $\phi \in \bigwedge V^{\star }$ and $x\in \bigwedge V:$%
\begin{equation*}
\phi \text{ is a }p\text{-homogeneous multiform, iff }\left[ \phi \right]
^{k}=o^{k},\text{ for }k\neq p
\end{equation*}%
\begin{equation}
x\text{ is a }p\text{-homogeneous multivector, iff }\left[ x\right]
_{k}=0_{k},\text{ for }k\neq p  \label{homogmulti}
\end{equation}%
i.e., $\phi =\left( 0,o^{1},o^{2},\ldots ,\phi ^{p},\ldots ,o^{n}\right) ,$ $%
x=\left( 0,0_{1},0_{2},\ldots ,x_{p},\ldots ,0_{n}\right) .$

We define the so-called inclusion operators for $p$-forms and $p$-vectors.

For any $\phi ^{p}\in \bigwedge^{p}V^{\star }$ and $x_{p}\in \bigwedge^{p}V:$%
\begin{equation*}
\left( \phi ^{p}\right) \in \bigwedge V^{\star }\text{ such that }\left[
\left( \phi ^{p}\right) \right] ^{k}=\left\{ 
\begin{array}{ll}
\phi ^{p}, & \text{for }k=p \\ 
o^{k}, & \text{for }k\neq p%
\end{array}%
\right.
\end{equation*}%
\begin{equation}
\left( x_{p}\right) \in \bigwedge V\text{ such that }\left[ \left(
x_{p}\right) \right] ^{k}=\left\{ 
\begin{array}{ll}
x_{p}, & \text{for }k=p \\ 
0_{k}, & \text{for }k\neq p%
\end{array}%
\right.  \label{inclusion}
\end{equation}%
i.e., $\left( \phi ^{p}\right) $ is a $p$-homogeneous multiform whose unique
non-necessarily null component is $\phi ^{p},$ and $\left( x_{p}\right) $ is
a $p$-homogeneous multivector whose unique non-necessarily null component is 
$x_{p}.$

Notice that the inclusion operators are linear mappings: from $%
\bigwedge^{p}V^{\star }$ to $\bigwedge V^{\star }$ or from $\bigwedge^{p}V$
to $\bigwedge V.$

The basic relations among $\left[ \left. {}\right. \right] ^{k},\left[
\left. {}\right. \right] _{k}$ and $\left( \left. {}\right. \right) $ are
the following:%
\begin{equation*}
\left( \left[ \phi \right] ^{p}\right) =\left( 0,o^{1},o^{2},\ldots ,\phi
^{p},\ldots ,o^{n}\right) ,\quad \left( \left[ x\right] _{p}\right) =\left(
0,0_{1},0_{2},\ldots ,x_{p},\ldots ,0_{n}\right)
\end{equation*}%
\begin{equation}
\left[ \left( \phi ^{p}\right) \right] ^{k}=\left\{ 
\begin{array}{ll}
\phi ^{p}, & \text{for }k=p \\ 
o^{k}, & \text{for }k\neq p%
\end{array}%
\right. ,\qquad \left[ \left( x_{p}\right) \right] _{k}=\left\{ 
\begin{array}{ll}
x_{p}, & \text{for }k=p \\ 
0_{p}, & \text{for }k\neq p%
\end{array}%
\right.  \label{relationship}
\end{equation}

We basically call extensor over $V$ any multilinear application from a
direct product of $p$-vectors spaces (or, $p$-forms spaces) to another
direct product of $q$-vectors spaces (or, $q$-forms spaces). The set of any
type of extensor is naturally a real vector space of finite dimension. In
this paper, we shall work \ only with some types of extensors, e.g., $\func{%
ext}\left( V;V^{\star }\right) ,$ whose elements are called crossed extensor
over $V,$ $\func{ext}\left( \bigwedge^{p}V;\bigwedge^{p}V^{\star }\right) $
and $\func{ext}\left( \bigwedge V;\bigwedge V^{\star }\right) ,$ etc.

\section{Duality pairing mapping}

We introduce four mappings which will be called by the name of duality
pairings. The first one applies pairs $\left( p\text{-form},p\text{-vector}%
\right) $ into scalars. It will be called duality pairing of $p$-forms with $%
p$-vectors. The second one is a kind of extension of the first one, it
applies pairs $\left( \text{multiform},\text{multivector}\right) $ into
scalars. It will be called duality pairing of multiforms with multivectors.
The third one applies pairs $\left( p\text{-vector},p\text{-form}\right) $
into scalars. It will be called duality pairing of $p$-vectors with $p$%
-forms. The fourth one is a kind of extension of the third one, it applies
pairs $\left( \text{multivector},\text{multiform}\right) $ into scalars. It
will be called duality pairing of multivectors with multiforms.

The pairing of $\phi ^{p}\in \dbigwedge^{p}V^{\star }$ with $x_{p}\in
\dbigwedge^{p}V$ is defined by

\begin{equation}
\left\langle \phi ^{p},x_{p}\right\rangle _{I}=\left\{ 
\begin{array}{ll}
\phi ^{0}x_{0}, & p=0 \\ 
\phi ^{1}(x_{1}), & p=1 \\ 
\dfrac{1}{p!}\phi ^{p}(e_{j_{1}},\ldots ,e_{j_{p}})x_{p}(\varepsilon
^{j_{1}},\ldots ,\varepsilon ^{j_{p}}), & 2\leq p\leq n%
\end{array}%
\right.  \label{pairingI}
\end{equation}%
where $\left\{ e_{j}\right\} $ is any basis of $V$ and $\left\{ \varepsilon
^{j}\right\} $ is its dual basis for $V^{\star }$ (i.e., $\varepsilon
^{k}(e_{j})=\delta _{j}^{k}$), and sums over the indices $j_{1},\ldots
,j_{p} $ from $1$ to $n$ are implicated.

We emphasize that $\phi ^{p}(e_{j_{1}},\ldots ,e_{j_{p}})x_{p}(\varepsilon
^{j_{1}},\ldots ,\varepsilon ^{j_{p}})$ is a real number which does not
depend on the pair of bases $\left( \left\{ e_{j}\right\} ,\left\{
\varepsilon ^{j}\right\} \right) $ used for calculating it, since both $\phi
^{p}$ and $x_{p}$ are multilinear applications. Therefore, $\left\langle
\phi ^{p},x_{p}\right\rangle _{I}$ is a right invariant scalar which depends
only on $\phi ^{p}$ and $x_{p}.$ Hence, $\left\langle \left. {}\right.
,\left. {}\right. \right\rangle _{I}$ is well defined.

The pairing of $\phi \in \bigwedge V^{\star }$ with $x\in \bigwedge V$ is
defined by%
\begin{equation}
\left\langle \phi ,x\right\rangle _{II}=\underset{k=0}{\overset{n}{\sum }}%
\left\langle \phi ^{k},x_{k}\right\rangle _{I}  \label{pairingII}
\end{equation}

The pairing of $x_{p}\in \bigwedge^{p}V$ with $\phi ^{p}\in
\dbigwedge^{p}V^{\star }$ is defined by%
\begin{equation}
\left\langle x_{p},\phi ^{p}\right\rangle _{III}=\left\langle \phi
^{p},x_{p}\right\rangle _{I}  \label{pairingIII}
\end{equation}

The pairing of $x\in \bigwedge V$ with $\phi \in \dbigwedge V^{\star }$ is
defined by%
\begin{equation}
\left\langle x,\phi \right\rangle _{IV}=\underset{k=0}{\overset{n}{\sum }}%
\left\langle x_{k},\phi ^{k}\right\rangle _{III}  \label{pairingIV}
\end{equation}

From (\ref{pairingIII}) it follows that for all $\phi ^{p}$ and $%
x_{p}:\left\langle \phi ^{p},x_{p}\right\rangle _{I}=\left\langle x_{p},\phi
^{p}\right\rangle _{III},$ and from (\ref{pairingII}), (\ref{pairingIII}), (%
\ref{pairingIV}) it follows that for all $\phi $ and $x:$ $\left\langle \phi
,x\right\rangle _{II}=\left\langle x,\phi \right\rangle _{IV}.$ Hence,
because of no confusion could arise, we shall use the same notation $%
\left\langle \left. {}\right. ,\left. {}\right. \right\rangle $ for all the
duality pairing mappings and we shall only speak about a single duality
pairing mapping.

The duality pairing mapping has two remarkable properties: the bilinearity
and the non-degeneracy.

For all $\alpha ,\beta \in \mathbb{R},$ $\phi ,\psi \in \dbigwedge V^{\star
} $ and $x,y\in \dbigwedge V:$%
\begin{equation*}
\left\langle \alpha \phi +\beta \psi ,x\right\rangle =\alpha \left\langle
\phi ,x\right\rangle +\beta \left\langle \psi ,x\right\rangle \qquad
\left\langle \phi ,\alpha x+\beta y\right\rangle =\alpha \left\langle \phi
,x\right\rangle +\beta \left\langle \phi ,y\right\rangle
\end{equation*}

Let $\phi \in \dbigwedge V^{\star }$ and $x\in \dbigwedge V:$%
\begin{equation*}
\lbrack \left\langle \phi ,x\right\rangle =0\text{ for all }\phi
]\Longrightarrow x=0\qquad \lbrack \left\langle \phi ,x\right\rangle =0\text{
for all }x]\Longrightarrow \phi =o
\end{equation*}

\section{Duality contraction mappings}

\subsection{Duality left contraction mapping}

We introduce four mappings which will be called by the name of duality left
contracting. The first one applies pairs $\left( p\text{-form},q\text{-vector%
}\right) ,$ with $p\leq q,$ into $\left( q-p\right) $-vectors, in its
definition it appears the duality pairing mapping. It will be called duality
left contraction of $q$-vectors by $p$-forms. The second one is an extension
of the first one, it applies pairs $\left( \text{multiform},\text{multivector%
}\right) $ into multivectors, and will be called duality left contraction of
multivectors by multiforms. The third one applies pairs $\left( p\text{%
-vector},q\text{-form}\right) ,$ with $p\leq q,$ into $\left( q-p\right) $%
-forms, in its definition it appears also the duality pairing mapping. It
will be called duality left contraction of $q$-forms by $p$-vectors. The
fourth one is an extension of the third one, it applies pairs $\left( \text{%
multivector},\text{multiform}\right) $ into multiforms, and will be called
duality left contraction of multiforms by multivectors.

The left contraction of $x_{q}\in \dbigwedge^{q}V$ by $\phi ^{p}\in
\dbigwedge^{p}V^{\star }$ $\left( p\leq q\right) $ is defined by%
\begin{equation}
\left\langle \phi ^{p},x_{q}\right\vert _{I}=\left\{ 
\begin{array}{ll}
\dfrac{1}{(q-p)!}\left\langle \widetilde{\phi }^{p}\wedge \varepsilon
^{j_{1}}\wedge \cdots \varepsilon ^{j_{q-p}},x_{q}\right\rangle
e_{j_{1}}\wedge \cdots e_{j_{q-p}}, & p<q \\ 
\left\langle \widetilde{\phi }^{p},x_{p}\right\rangle , & p=q%
\end{array}%
\right.  \label{lcontractionI}
\end{equation}%
where $\left\{ e_{j}\right\} $ is any basis of $V$ and $\left\{ \varepsilon
^{j}\right\} $ is its dual basis for $V^{\star }$ (i.e., $\varepsilon
^{k}(e_{j})=\delta _{j}^{k}$), and sums over the indices $j_{1},\ldots
,j_{q-p}$ from $1$ to $n$ are implicated.

We emphasize that $\left\langle \widetilde{\phi }^{p}\wedge \varepsilon
^{j_{1}}\wedge \cdots \varepsilon ^{j_{q-p}},x_{q}\right\rangle
e_{j_{1}}\wedge \cdots e_{j_{q-p}}$ is a $\left( q-p\right) $-vector which
does not depend on the pair of bases $\left( \left\{ e_{j}\right\} ,\left\{
\varepsilon ^{j}\right\} \right) $ used for defining it, since both $\phi
^{p}$ and $x_{q}$ are multilinear applications. Therefore, $\left\langle
\phi ^{p},x_{q}\right\vert _{I}$ is a right invariant $\left( q-p\right) $%
-vector which depends only on $\phi ^{p}$ and $x_{q}.$ Hence, $\left\langle
\left. {}\right. ,\left. {}\right. \right\vert _{I}$ is well defined.

The left contraction of $x\in \dbigwedge V$ by $\phi \in \dbigwedge V^{\star
}$ is defined by%
\begin{equation}
\left\langle \phi ,x\right\vert _{II}\in \bigwedge V\text{ such that }\left[
\left\langle \phi ,x\right\vert _{II}\right] _{k}=\sum_{j=0}^{n-k}\left%
\langle \phi ^{j},x_{j+k}\right\vert _{I},\text{ }\left( 0\leq k\leq n\right)
\label{lcontractionII}
\end{equation}

The left contraction of $\phi ^{q}\in \bigwedge^{q}V^{\star }$ by $x_{p}\in
\dbigwedge^{p}V$ $\left( p\leq q\right) $ is defined by%
\begin{equation}
\left\langle x_{p},\phi ^{q}\right\vert _{III}=\left\{ 
\begin{array}{ll}
\dfrac{1}{(q-p)!}\left\langle \widetilde{x}_{p}\wedge e_{j_{1}}\wedge \cdots
e_{j_{q-p}},\phi ^{q}\right\rangle \varepsilon ^{j_{1}}\wedge \cdots
\varepsilon ^{j_{q-p}}, & p<q \\ 
\left\langle \widetilde{x}_{p},\phi ^{p}\right\rangle , & p=q%
\end{array}%
\right.  \label{lcontractionIII}
\end{equation}%
where sums over the indices $j_{1},\ldots ,j_{q-p}$ from $1$ to $n$ are
implicated. Analogous remarks to those just done about $\left\langle \left.
{}\right. ,\left. {}\right. \right\vert _{I}$ allow us to assert that $%
\left\langle \left. {}\right. ,\left. {}\right. \right\vert _{III}$ is also
well defined.

The left contraction of $\phi \in \dbigwedge V^{\star }$ by $x\in \dbigwedge
V$ is defined by%
\begin{equation}
\left\langle x,\phi \right\vert _{IV}\in \bigwedge V^{\star }\text{ such
that }\left[ \left\langle x,\phi \right\vert _{IV}\right] ^{k}=%
\sum_{j=0}^{n-k}\left\langle x_{j},\phi ^{j+k}\right\vert _{III},\text{ }%
\left( 0\leq k\leq n\right)  \label{lcontractionIV}
\end{equation}

Because of no confusion could arise, we shall use the same notation $%
\left\langle \left. {}\right. ,\left. {}\right. \right\vert $ for all of the
duality left contraction mappings and we will only speak about a single
duality left contraction mapping.

The duality left contraction mapping has two remarkable properties: the
bilinearity and the non-degeneracy.

For all $\alpha ,\beta \in \mathbb{R},$ $\phi ,\psi \in \dbigwedge V^{\star
} $ and $x,y\in \dbigwedge V:$%
\begin{eqnarray*}
\left\langle \alpha \phi +\beta \psi ,x\right\vert &=&\alpha \left\langle
\phi ,x\right\vert +\beta \left\langle \psi ,x\right\vert \quad \left\langle
\phi ,\alpha x+\beta y\right\vert =\alpha \left\langle \phi ,x\right\vert
+\beta \left\langle \phi ,y\right\vert \\
\left\langle \alpha x+\beta y,\phi \right\vert &=&\alpha \left\langle x,\phi
\right\vert +\beta \left\langle y,\phi \right\vert \quad \left\langle
x,\alpha \phi +\beta \psi \right\vert =\alpha \left\langle x,\phi
\right\vert +\beta \left\langle x,\psi \right\vert
\end{eqnarray*}

Let $\phi \in \dbigwedge V^{\star }$ and $x\in \dbigwedge V:$%
\begin{eqnarray*}
\lbrack \left\langle \phi ,x\right\vert &=&0\text{ for all }\phi
]\Longrightarrow x=0\qquad \lbrack \left\langle \phi ,x\right\vert =0\text{
for all }x]\Longrightarrow \phi =o \\
\lbrack \left\langle x,\phi \right\vert &=&o\text{ for all }%
x]\Longrightarrow \phi =o\qquad \lbrack \left\langle x,\phi \right\vert =o%
\text{ for all }\phi ]\Longrightarrow x=0
\end{eqnarray*}

\subsection{Duality right contraction mapping}

We introduce now other four mappings else which will be called by the name
of duality right contraction. In definitions of two of them, it appears also
the duality pairing mapping. The first one applies pairs $\left( q\text{%
-vector},p\text{-form}\right) ,$ with $p\leq q,$ into $\left( q-p\right) $%
-vectors, and will be called duality right contraction of $q$-vectors by $p$%
-forms. The second one is an extension of the first one, it applies pairs $%
\left( \text{multivector},\text{multiform}\right) $ into multivectors. It
will be called duality right contraction of multivectors by multiforms. The
third one applies pairs $\left( q\text{-form},p\text{-vector}\right) ,$ with 
$p\leq q,$ into $\left( q-p\right) $-forms, and will be called duality right
contraction of $q$-forms by $p$-vectors. The fourth one is an extension of
the third one, it applies pairs $\left( \text{multiform},\text{multivector}%
\right) $ into multiforms. It will be called duality right contraction of
multiforms by multivectors.

The right contraction of $x_{q}\in \dbigwedge^{q}V$ by $\phi ^{p}\in
\dbigwedge^{p}V^{\star }$ $\left( p\leq q\right) $ is defined by%
\begin{equation}
\left\vert x_{q},\phi ^{p}\right\rangle _{I}=\left\{ 
\begin{array}{ll}
\dfrac{1}{(q-p)!}\left\langle x_{q},\varepsilon ^{j_{1}}\wedge \cdots
\varepsilon ^{j_{q-p}}\wedge \widetilde{\phi }^{p}\right\rangle
e_{j_{1}}\wedge \cdots e_{j_{q-p}}, & p<q \\ 
\left\langle x_{p},\widetilde{\phi }^{p}\right\rangle , & p=q%
\end{array}%
\right.  \label{rcontractionI}
\end{equation}%
where $\left\{ e_{j}\right\} $ is any basis of $V$ and $\left\{ \varepsilon
^{j}\right\} $ is its dual basis for $V^{\star }$ (i.e., $\varepsilon
^{k}(e_{j})=\delta _{j}^{k}$), and sums over the indices $j_{1},\ldots
,j_{q-p}$ from $1$ to $n$ are implicated.

We emphasize that $\left\langle x_{q},\varepsilon ^{j_{1}}\wedge \cdots
\varepsilon ^{j_{q-p}}\wedge \widetilde{\phi }^{p}\right\rangle
e_{j_{1}}\wedge \cdots e_{j_{q-p}}$ is a $\left( q-p\right) $-vector which
does not depend on the pair of bases $\left( \left\{ e_{j}\right\} ,\left\{
\varepsilon ^{j}\right\} \right) $ used for defining it, since both $\phi
^{p}$ and $x_{q}$ are multilinear applications. Therefore, $\left\vert
x_{q},\phi ^{p}\right\rangle _{I}$ is a right invariant $\left( q-p\right) $%
-vector which depends only on $\phi ^{p}$ and $x_{q}.$ Hence, $\left\vert
\left. {}\right. ,\left. {}\right. \right\rangle _{I}$ is well defined.

The right contraction of $x\in \dbigwedge V$ by $\phi \in \dbigwedge
V^{\star }$ is defined by%
\begin{equation}
\left\vert x,\phi \right\rangle _{II}\in \dbigwedge V\text{ such that }\left[
\left\vert x,\phi \right\rangle _{II}\right] =\sum_{j=0}^{n-k}\left\vert
x_{j+k},\phi ^{j}\right\rangle _{I},\text{ }\left( 0\leq k\leq n\right)
\label{rcontractionII}
\end{equation}

The right contraction of $\phi ^{q}\in \dbigwedge^{q}V^{\star }$ by $%
x_{p}\in \dbigwedge^{p}V$ $\left( p\leq q\right) $ is defined by%
\begin{equation}
\left\vert \phi ^{q},x_{p}\right\rangle _{III}=\left\{ 
\begin{array}{ll}
\dfrac{1}{(q-p)!}\left\langle \phi ^{q},e_{j_{1}}\wedge \cdots
e_{j_{q-p}}\wedge \widetilde{x}_{p}\right\rangle \varepsilon ^{j_{1}}\wedge
\cdots \varepsilon ^{j_{q-p}}, & p<q \\ 
\left\langle \phi ^{q},\widetilde{x}_{p}\right\rangle , & p=q%
\end{array}%
\right.  \label{rcontractionIII}
\end{equation}%
where sums over the indices $j_{1},\ldots ,j_{q-p}$ from $1$ to $n$ are
implicated. Similar remarks to those just done about $\left\vert \left.
{}\right. ,\left. {}\right. \right\rangle _{I}$ allow us to affirm that $%
\left\vert \left. {}\right. ,\left. {}\right. \right\rangle _{III}$ is also
well defined.

The right contraction of $\phi \in \dbigwedge V^{\star }$ by $x\in
\dbigwedge V$ is defined by%
\begin{equation}
\left\vert \phi ,x\right\rangle _{IV}\in \dbigwedge V^{\star }\text{ such
that }\left[ \left\vert \phi ,x\right\rangle _{IV}\right] =\sum_{j=0}^{n-k}%
\left\vert \phi ^{j+k},x_{j}\right\rangle _{III},\text{ }\left( 0\leq k\leq
n\right)  \label{rcontractionIV}
\end{equation}

Because of no confusion could arise, we shall use the same notation $%
\left\vert \left. {}\right. ,\left. {}\right. \right\rangle $ for all of the
duality right contraction mappings and we shall only speak about a single
duality right contraction mapping.

The duality right contraction mapping has two remarkable properties: the
bilinearity and the non-degeneracy.

For all $\alpha ,\beta \in \mathbb{R},$ $\phi ,\psi \in \dbigwedge V^{\star
} $ and $x,y\in \dbigwedge V:$%
\begin{eqnarray*}
\left\vert \alpha x+\beta y,\phi \right\rangle &=&\alpha \left\vert x,\phi
\right\rangle +\beta \left\vert y,\phi \right\rangle \text{\quad }\left\vert
x,\alpha \phi +\beta \psi \right\rangle =\alpha \left\vert x,\phi
\right\rangle +\beta \left\vert x,\psi \right\rangle \\
\left\vert \alpha \phi +\beta \psi ,x\right\rangle &=&\alpha \left\vert \phi
,x\right\rangle +\beta \left\vert \psi ,x\right\rangle \text{\quad }%
\left\vert \phi ,\alpha x+\beta y\right\rangle =\alpha \left\vert \phi
,x\right\rangle +\beta \left\vert \phi ,y\right\rangle
\end{eqnarray*}

Let $\phi \in \dbigwedge V^{\star }$ and $x\in \dbigwedge V:$%
\begin{eqnarray*}
\lbrack \left\vert x,\phi \right\rangle &=&0\text{ for all }\phi
]\Longrightarrow x=0\text{\qquad }[\left\vert x,\phi \right\rangle =0\text{
for all }x]\Longrightarrow \phi =o \\
\lbrack \left\vert \phi ,x\right\rangle &=&o\text{ for all }%
x]\Longrightarrow \phi =o\text{\qquad }[\left\vert \phi ,x\right\rangle =o%
\text{ for all }\phi ]\Longrightarrow x=0
\end{eqnarray*}

\section{Identities involving the duality mappings}

We present a list of identities which involve the duality mappings.

For all $\phi ^{p}\in \dbigwedge^{p}V^{\star }$ and $v_{1},v_{2},\ldots
,v_{p}\in V:$%
\begin{equation}
\left\langle \phi ^{p},v_{1}\wedge v_{2}\wedge \cdots \wedge
v_{p}\right\rangle =\phi ^{p}(v_{1},v_{2},\ldots ,v_{p})  \label{Fund1}
\end{equation}

For all $x_{p}\in \dbigwedge^{p}V$ and $\omega _{1},\omega _{2},\ldots
,\omega _{p}\in V^{\star }$ $(p\geq 2):$%
\begin{equation}
\left\langle \omega _{1}\wedge \omega _{2}\wedge \cdots \wedge \omega
_{p},x_{p}\right\rangle =x_{p}(\omega _{1},\omega _{2},\ldots ,\omega _{p})
\label{Fund2}
\end{equation}

For all $\omega _{1},\ldots ,\omega _{p}\in V^{\star }$ and $v_{1},\ldots
,v_{p}\in V:$%
\begin{equation}
\left\langle \omega _{1}\wedge \cdots \wedge \omega _{p},v_{1}\wedge \cdots
\wedge v_{p}\right\rangle =\det \left( \left\langle \omega
_{i},v_{j}\right\rangle \right)  \label{Fund3}
\end{equation}

For all $\phi ^{p}\in \bigwedge^{p}V^{\star },$ $x_{q}\in \bigwedge^{q}V,$
and $y_{p}\in \bigwedge^{p}V,$ $\psi ^{q}\in \bigwedge^{q}V^{\star }$ $%
\left( p\leq q\right) :$%
\begin{equation}
\left\langle \phi ^{p},x_{q}\right\vert =\left( -1\right)
^{p(q-p)}\left\vert x_{q},\phi ^{p}\right\rangle ,\quad \left\langle
y_{p},\psi ^{q}\right\vert =\left( -1\right) ^{p(q-p)}\left\vert y_{q},\psi
^{p}\right\rangle .  \label{Fund4}
\end{equation}

For all $\omega \in V^{\star },$ $v_{1},v_{2},\ldots ,v_{p}\in V,$ and $v\in
V,$ $\omega _{1},\omega _{2},\ldots ,\omega _{p}\in V^{\star }:$%
\begin{eqnarray}
\left\langle \omega ,v_{1}\wedge v_{2}\wedge \cdots v_{p}\right\vert
&=&\sum_{k-1}^{p}(-1)^{k-1}\left\langle \omega ,v_{k}\right\rangle
v_{1}\wedge v_{2}\wedge \cdots \overset{\curlyvee }{v}_{k}\cdots v_{p},
\label{contrformvectors} \\
\left\langle v,\omega _{1}\wedge \omega _{2}\wedge \cdots \omega
_{p}\right\vert &=&\sum_{k-1}^{p}(-1)^{k-1}\left\langle v,\omega
_{k}\right\rangle \omega _{1}\wedge \omega _{2}\wedge \cdots \overset{%
\curlyvee }{\omega }_{k}\cdots \omega _{p}.  \label{contrvectorforms}
\end{eqnarray}%
where $\overset{\curlyvee }{v}_{k}$ means that $v_{k}$ must be removed from $%
v_{1}\wedge v_{2}\wedge \cdots v_{p}$ but leaving the remaining vectors in
the order they had, and analogously for $\overset{\curlyvee }{\omega }_{k}.$

For all $\left( \phi ^{p}\right) \in \bigwedge V^{\star }$ and $\left(
x_{p}\right) \in \bigwedge V:$%
\begin{equation}
\left\langle \left( \phi ^{p}\right) ,\left( x_{p}\right) \right\rangle
=\left\langle \phi ^{p},x_{p}\right\rangle  \label{scalarhomog1}
\end{equation}

For all $\left( \phi ^{p}\right) \in \bigwedge V^{\star }$ and $\left(
x_{q}\right) \in \bigwedge V:$%
\begin{equation}
\left\langle \left( \phi ^{p}\right) ,\left( x_{q}\right) \right\rangle =0,%
\text{ if }p\neq q  \label{scalarhomog2}
\end{equation}

For all $\left( \phi ^{p}\right) \in \bigwedge V^{\star }$ and $\left(
x_{q}\right) \in \bigwedge V:$%
\begin{equation}
\left\langle \left( \phi ^{p}\right) ,\left( x_{q}\right) \right\vert
,\left\vert \left( x_{q}\right) ,\left( \phi ^{p}\right) \right\rangle \text{%
are }\left( q-p\right) \text{-homogeneous multivector, and}
\label{contrhomogmultiv0}
\end{equation}%
\begin{equation*}
\left[ \left\langle \left( \phi ^{p}\right) ,\left( x_{q}\right) \right\vert %
\right] _{q-p}=\left\langle \phi ^{p},x_{q}\right\vert ,\text{ }\left[
\left\vert \left( x_{q}\right) ,\left( \phi ^{p}\right) \right\rangle \right]
_{q-p}=\left\vert x_{q},\phi ^{p}\right\rangle ,\text{ if }p\leq q
\end{equation*}%
\begin{equation}
\left\langle \left( \phi ^{p}\right) ,\left( x_{q}\right) \right\vert
=\left\vert \left( x_{q}\right) ,\left( \phi ^{p}\right) \right\rangle =0,%
\text{ if }p>q  \label{contrhomogmultiv}
\end{equation}

For all $\left( x_{p}\right) \in \bigwedge V$ and $\left( \phi ^{q}\right)
\in \bigwedge V^{\star }:$%
\begin{equation}
\left\langle \left( x_{p}\right) ,\left( \phi ^{q}\right) \right\vert ,\text{
}\left\vert \left( \phi ^{q}\right) ,\left( x_{p}\right) \right\rangle \text{
are }\left( q-p\right) \text{-homogeneous multiform, and}
\label{contrhomogmultif0}
\end{equation}%
\begin{equation*}
\left[ \left\langle \left( x_{p}\right) ,\left( \phi ^{q}\right) \right\vert %
\right] _{q-p}=\left\langle x_{p},\phi ^{q}\right\vert ,\text{ }\left[
\left\vert \left( \phi ^{q}\right) ,\left( x_{p}\right) \right\rangle \right]
_{q-p}=\left\vert \phi ^{q},x_{p}\right\rangle ,\text{if }p\leq q
\end{equation*}%
\begin{equation}
\left\langle \left( x_{p}\right) ,\left( \phi ^{q}\right) \right\vert
=\left\vert \left( \phi ^{q}\right) ,\left( x_{q}\right) \right\rangle =o,%
\text{ if }p>q  \label{contrhomogmultif}
\end{equation}

For all $\omega \in V^{\star }$ and $x,y\in \bigwedge V:$%
\begin{equation}
\left\langle \omega ,x\wedge y\right\vert =\left\langle \omega ,x\right\vert
\wedge y+\widehat{x}\wedge \left\langle \omega ,y\right\vert
\label{contrformmultivectors}
\end{equation}

For all $v\in V$ and $\phi ,\psi \in \bigwedge V^{\star }:$%
\begin{equation}
\left\langle v,\phi \wedge \psi \right\vert =\left\langle v,\phi \right\vert
\wedge \psi +\widehat{\phi }\wedge \left\langle v,\psi \right\vert
\label{contrvectormultiforms}
\end{equation}

For all $\phi ,\psi \in \bigwedge V^{\star }$ and $x,y\in \bigwedge V:$%
\begin{equation}
\left\langle \phi ,\left\langle \psi ,x\right\vert \right\vert =\left\langle
\phi \wedge \psi ,x\right\vert \hspace{0.5in}\left\vert \left\vert x,\phi
\right\rangle ,\psi \right\rangle =\left\vert x,\phi \wedge \psi
\right\rangle  \label{dcvectors}
\end{equation}%
\begin{equation}
\left\langle x,\left\langle y,\phi \right\vert \right\vert =\left\langle
x\wedge y,\phi \right\vert \hspace{0.5in}\left\vert \left\vert \phi
,x\right\rangle ,y\right\rangle =\left\vert \phi ,x\wedge y\right\rangle
\label{dcforms}
\end{equation}%
\begin{equation}
\left\langle \left\langle \phi ,x\right\vert ,\psi \right\rangle
=\left\langle x,\widetilde{\phi }\wedge \psi \right\rangle \hspace{0.5in}%
\left\langle \phi ,\left\vert x,\psi \right\rangle \right\rangle
=\left\langle \phi \wedge \widetilde{\psi },x\right\rangle
\label{dcpvectors}
\end{equation}%
\begin{equation}
\left\langle \left\langle x,\phi \right\vert ,y\right\rangle =\left\langle
\phi ,\widetilde{x}\wedge y\right\rangle \hspace{0.5in}\left\langle
x,\left\vert \phi ,y\right\rangle \right\rangle =\left\langle x\wedge 
\widetilde{y},\phi \right\rangle  \label{dcpforms}
\end{equation}

Let $\left\{ e_{j}\right\} $ be any basis of $V$ and let $\left\{
\varepsilon ^{j}\right\} $ be its dual basis for $V^{\star }.$ And let us
introduce the pseudoscalars $e_{\wedge }=e_{1}\wedge e_{2}\wedge \cdots
\wedge e_{n}$ and $\varepsilon ^{\wedge }=\varepsilon ^{1}\wedge \varepsilon
^{2}\wedge \cdots \wedge \varepsilon ^{n}.$

\begin{equation}
\left\langle \varepsilon ^{\wedge },e_{\wedge }\right\rangle =1
\label{pairingofpseudo}
\end{equation}

For all basis vectors $e_{p_{1}},e_{p_{2}},\ldots ,e_{p_{\mu }}$ $\left(
1\leq p_{1}<p_{2}<\cdots <p_{\mu }\leq n\right) :$%
\begin{equation}
\left\langle e_{p_{1}}\wedge e_{p_{2}}\wedge \cdots e_{p_{\nu }},\varepsilon
^{\wedge }\right\vert =\left( -1\right) ^{\mu +p_{1}+p_{2}+\cdots p_{\mu
}}\varepsilon ^{\wedge }\left\{ \varepsilon ^{p_{1}},\varepsilon
^{p_{2}};\ldots ,\varepsilon ^{p_{\nu }}\right\}  \label{expansion1}
\end{equation}%
where $\varepsilon ^{\wedge }\left\{ \varepsilon ^{p_{1}},\varepsilon
^{p_{2}};\ldots ,\varepsilon ^{p_{\mu }}\right\} $ is the remaining $\left(
n-\mu \right) $-form once the $\mu $ dual basis forms $\varepsilon
^{p_{1}},\varepsilon ^{p_{2}},\ldots ,\varepsilon ^{p_{\mu }}$ are removed
from $\varepsilon ^{\wedge }$ but leaving the others ones in the order they
had.

For all dual basis forms $\varepsilon ^{q_{1}},\varepsilon ^{q_{2}},\ldots
,\varepsilon ^{q_{\nu }}$ $\left( 1\leq q_{1}<q_{2}<\cdots <q_{\nu }\leq
n\right) :$%
\begin{equation}
\left\langle \varepsilon ^{q_{1}}\wedge \varepsilon ^{q_{2}}\wedge \cdots
\varepsilon ^{q_{\nu }},e_{\wedge }\right\vert =\left( -1\right) ^{\nu
+q_{1}+q_{2}+\cdots q_{\nu }}e_{\wedge }\left\{ e_{q_{1}},e_{q_{2}},\ldots
,e_{q_{\nu }}\right\}  \label{expansion2}
\end{equation}%
where $e_{\wedge }\left\{ e_{q_{1}},e_{q_{2}},\ldots ,e_{q_{\nu }}\right\} $
is the remaining $\left( n-\nu \right) $-vector once the $\nu $ basis
vectors $e_{q_{1}},e_{q_{2}},\ldots ,e_{q_{\nu }}$ are removed from $%
e_{\wedge }$ but leaving the others ones in the order they had.

For all $\alpha \in \mathbb{R}:$%
\begin{equation}
\left\langle \left\langle \alpha ,\varepsilon ^{\wedge }\right\vert ,%
\widetilde{e}_{\wedge }\right\vert =\alpha \hspace{0.5in}\left\langle
\left\langle \alpha ,e_{\wedge }\right\vert ,\widetilde{\varepsilon }%
^{\wedge }\right\vert =\alpha  \label{expansionformula0}
\end{equation}

For all $v\in V$ and $\omega \in V^{\star }:$%
\begin{equation}
\left\langle \left\langle v,\varepsilon ^{\wedge }\right\vert ,\widetilde{e}%
_{\wedge }\right\vert =v\hspace{0.5in}\left\langle \left\langle \omega
,e_{\wedge }\right\vert ,\widetilde{\varepsilon }^{\wedge }\right\vert
=\omega  \label{expansionformula1}
\end{equation}

For all $v_{1},v_{2},\ldots ,v_{p}\in V$ and $\omega _{1},\omega _{2},\ldots
,\omega _{p}\in V^{\star }:$%
\begin{eqnarray}
\left\langle \left\langle v_{1}\wedge v_{2}\wedge \cdots \wedge
v_{p},\varepsilon ^{\wedge }\right\vert ,\widetilde{e}_{\wedge }\right\vert
&=&v_{1}\wedge v_{2}\wedge \cdots \wedge v_{p}  \notag \\
\left\langle \left\langle \omega _{1}\wedge \omega _{2}\wedge \cdots \wedge
\omega _{p},e_{\wedge }\right\vert ,\widetilde{\varepsilon }^{\wedge
}\right\vert &=&\omega _{1}\wedge \omega _{2}\wedge \cdots \wedge \omega _{p}
\label{expansionformula2}
\end{eqnarray}

For all $x\in \bigwedge V$ and $\phi \in \bigwedge V^{\star }:$%
\begin{equation}
\left\langle \left\langle x,\varepsilon ^{\wedge }\right\vert ,\widetilde{e}%
_{\wedge }\right\vert =x\hspace{0.5in}\left\langle \left\langle \phi
,e_{\wedge }\right\vert ,\widetilde{\varepsilon }^{\wedge }\right\vert =\phi
\label{expansionformula3}
\end{equation}

Notice that (\ref{scalarhomog1}), (\ref{scalarhomog2}), (\ref%
{contrhomogmultiv0}), (\ref{contrhomogmultiv}), (\ref{contrhomogmultif0})
and (\ref{contrhomogmultif}) allow us to identify $p$-forms with $p$%
-homogeneous multiforms, and $p$-vectors with $p$-homogeneous multivectors,
whenever in doing calculations with pairing or contractions.

The equations (\ref{expansionformula1}), (\ref{expansionformula2}) and (\ref%
{expansionformula3}) are respectively expansion-like formulas for vectors
and forms, simple $p$-vectors and simple $p$-forms, and multivectors and
multiforms. All of them involve only the duality left contraction mapping.
Of course, there exist similar others which involve only the duality right
contraction mapping.

\section{Metric extensor}

A metric tensor over $V$ is any covariant $2$-tensor over $V,$ say $g,$
which is symmetric and non-degenerate, i.e.,%
\begin{equation*}
g(v,w)=g(w,v)
\end{equation*}%
\begin{equation*}
\lbrack g(v,w)=0\text{ for all }w]\text{ implies }v=0_{V}
\end{equation*}

A crossed extensor\footnote{%
Any linear application from $V$ to $V^{\star }$ is called crossed extensor
over $V,$ the set of crossed extensors is a real vector space of finite
dimension, and it is denoted by $\func{ext}\left( V;V^{\star }\right) .$}
over $V,$ say $\gamma ,$ which is symmetric and one-to-one, i.e.,%
\begin{equation*}
\left\langle \gamma (v),w\right\rangle =\left\langle v,\gamma
(w)\right\rangle
\end{equation*}%
\begin{equation*}
\gamma (v)=\gamma (v^{\prime })\text{ implies }v=v^{\prime }
\end{equation*}%
will be called metric extensor over $V.$

\begin{theorem}
For each metric tensor $g$ there exists an unique metric extensor $\gamma $
such that for all $v,w\in V:g(v,w)=\left\langle \gamma (v),w\right\rangle ,$
and reciprocally.
\end{theorem}

\begin{proof}
We first prove the direct statement:

Given $g\in T^{2}V^{\star },$ we must prove the existence of such $\gamma
\in \func{ext}\left( V;V^{\star }\right) .$ Let%
%TCIMACRO{\U{b4}}%
%BeginExpansion
\'{}%
%EndExpansion
s define $\gamma :V\longrightarrow V^{\star }$ such that $\gamma
(v)=g(v,e_{j})\varepsilon ^{j},$ where $\left\{ e_{j}\right\} $ is any basis
of $V$ and $\left\{ \varepsilon ^{j}\right\} $ is its dual basis for $%
V^{\star }$ (i.e., $\left\langle e_{j},\varepsilon ^{k}\right\rangle =\delta
_{j}^{k}$), and sum over $j$ from $1$ to $n$ is implicated. Since $g$ is
linear with respect to 2nd. variable, $g(v,e_{j})\varepsilon ^{j}$ is a form
that does not depend on the pair of bases $\left( \left\{ e_{j}\right\}
,\left\{ \varepsilon ^{j}\right\} \right) $ used for defining it. Hence, $%
\gamma $ is an application well defined. And, since $g$ is linear with
respect to 1st. variable, $\gamma $ is linear. Therefore, $\gamma \in \func{%
ext}\left( V;V^{\star }\right) .$

We prove that $\gamma $ is symmetric. Let $v,w\in V:$%
\begin{eqnarray*}
\left\langle \gamma (v),w\right\rangle &=&\left\langle g(v,e_{j})\varepsilon
^{j},w\right\rangle =g(v,e_{j})\left\langle \varepsilon ^{j},w\right\rangle
=g(v,\left\langle \varepsilon ^{j},w\right\rangle e_{j})=g(v,w), \\
\left\langle v,\gamma (w)\right\rangle &=&\left\langle
v,g(w,e_{j})\varepsilon ^{j}\right\rangle =g(w,e_{j})\left\langle
v,\varepsilon ^{j}\right\rangle =g(w,\left\langle v,\varepsilon
^{j}\right\rangle e_{j})=g(w,v),
\end{eqnarray*}%
then, since $g$ is symmetric, $\left\langle \gamma (v),w\right\rangle
=\left\langle v,\gamma (w)\right\rangle .$

We prove that $\gamma $ is one-to-one. Take $v,v^{\prime },w\in V:$%
\begin{eqnarray*}
\gamma (v)=\gamma (v^{\prime }) &\Longrightarrow &\left\langle \gamma
(v),w\right\rangle =\left\langle \gamma (v^{\prime }),w\right\rangle
\Longrightarrow g(v,w)=g(v^{\prime },w) \\
&\Longrightarrow &g(v-v^{\prime },w)=0\text{ for all }w
\end{eqnarray*}%
then, since $g$ is non-degenerate, $v-v^{\prime }=0_{V}.$ Thus, $\gamma
(v)=\gamma (v^{\prime })$ implies $v=v^{\prime }.$

In order to prove the uniqueness of such $\gamma ,$ suppose that there
exists another $\gamma ^{\prime }\in \func{ext}\left( V;V^{\star }\right) $
which satisfies $\left\langle \gamma ^{\prime }(v),w\right\rangle =g(v,w).$
Since $g(v,w)=\left\langle \gamma (v),w\right\rangle ,$ we have $%
\left\langle \gamma ^{\prime }(v),w\right\rangle =\left\langle \gamma
(v),w\right\rangle ,$ but%
\begin{equation*}
\left\langle \gamma ^{\prime }(v),w\right\rangle =\left\langle \gamma
(v),w\right\rangle \Longrightarrow \left\langle \gamma ^{\prime }(v)-\gamma
(v),w\right\rangle =0\text{ for all }w
\end{equation*}%
then, because of non-degeneracy of $\left\langle \left. {}\right. ,\left.
{}\right. \right\rangle ,$ we have $\gamma ^{\prime }(v)-\gamma
(v)=0_{V^{\star }},$ i.e., $\gamma ^{\prime }(v)=\gamma (v).$

We now prove the reciprocal statement:

Given $\gamma \in \func{ext}\left( V;V^{\star }\right) ,$ we must prove the
existence of such $g\in T^{2}V^{\star }.$ Let%
%TCIMACRO{\U{b4}}%
%BeginExpansion
\'{}%
%EndExpansion
s define $g:V\times V\longrightarrow \mathbb{R}$ such that $%
g(v,w)=\left\langle \gamma (v),w\right\rangle .$ Because of linearity of $%
\gamma $ and bilinearity of $\left\langle \left. {}\right. ,\left. {}\right.
\right\rangle ,$ it follows the bilinearity of $g.$

We prove that $g$ is symmetric. Let $v,w\in V:$%
\begin{equation*}
g(v,w)=\left\langle \gamma (v),w\right\rangle ,\text{\quad }%
g(w,v)=\left\langle \gamma (w),v\right\rangle =\left\langle v,\gamma
(w)\right\rangle
\end{equation*}%
then, since $\gamma $ is symmetric, $g(v,w)=g(w,v).$

We prove that $g$ is non-degenerate. Take $v,v^{\prime },w\in V:$%
\begin{equation*}
g(v,w)=0\text{ for all }w\Longrightarrow \left\langle \gamma
(v),w\right\rangle =0\text{ for all }w
\end{equation*}%
then, because of non-degeneracy of $\left\langle \left. {}\right. ,\left.
{}\right. \right\rangle $ and linearity of $\gamma ,$ we have $\gamma
(v)=o_{V^{\star }}=\gamma (0_{V}).$ And, since $\gamma $ is one-to-one, $%
v=0_{V}.$ Thus, $[g(v,w)=0$ for all $w]$ implies $v=0_{V}.$

In order to prove the uniqueness of such $g,$ suppose that there exists
another $g^{\prime }\in T^{2}V^{\star }$ which satisfies $g^{\prime
}(v,w)=\left\langle \gamma (v),w\right\rangle .$ Since $g(v,w)=\left\langle
\gamma (v),w\right\rangle ,$ we have $g^{\prime }(v,w)=g(v,w).$
\end{proof}

\begin{proposition}
$\gamma $ is invertible, and $\gamma ^{-1}$ is a metric extensor over $%
V^{\star }.$
\end{proposition}

\begin{proof}
We first prove that $\gamma $ is invertible:

We only must check that $\gamma $ is onto. Since $\gamma $ is one-to-one, $%
\ker \gamma =\left\{ 0_{1}\right\} .$ And, since $\dim \left( \ker \gamma
\right) +\dim \left( \limfunc{im}\gamma \right) =\dim V,$ we have $\dim
\left( \limfunc{im}\gamma \right) =n,$ i.e., $\gamma $ is onto.

We now prove that $\gamma ^{-1}$ is a metric extensor over $V^{\star }:$

We only must check that $\gamma ^{-1}$ is symmetric because $\gamma ^{-1}$
is one-to-one too.

Let $\omega ,\sigma \in V^{\star },$ there exist $v,w\in V$ such that $%
\omega =\gamma (v)$ and $\sigma =\gamma (w)$ or, equivalently, $v=\gamma
^{-1}(\omega )$ and $w=\gamma ^{-1}(\sigma ).$ We can write%
\begin{equation*}
\left\langle \gamma ^{-1}(\omega ),\sigma \right\rangle =\left\langle
v,\gamma (w)\right\rangle \text{ and }\left\langle \omega ,\gamma
^{-1}(\sigma )\right\rangle =\left\langle \gamma (v),w\right\rangle
\end{equation*}%
then, since $\gamma $ is symmetric, $\left\langle \gamma ^{-1}(\omega
),\sigma \right\rangle =\left\langle \omega ,\gamma ^{-1}(\sigma
)\right\rangle .$
\end{proof}

\subsection{Extension of the metric extensor}

We define the extensions of $\gamma $ to $\func{ext}\left(
\bigwedge^{p}V;\bigwedge^{p}V^{\star }\right) $ and to $\func{ext}\left(
\bigwedge V;\bigwedge V^{\star }\right) .$

The extension of $\gamma $ to $\func{ext}\left(
\bigwedge^{p}V;\bigwedge^{p}V^{\star }\right) $ is $\underline{\gamma }%
^{I}:\bigwedge^{p}V\longrightarrow \bigwedge^{p}V^{\star }$ defined by%
\begin{equation}
\underline{\gamma }^{I}(x_{p})=\left\{ 
\begin{array}{ll}
x_{0}, & p=0 \\ 
\gamma (x_{1}), & p=1 \\ 
\dfrac{1}{p!}\left\langle x_{p},\varepsilon ^{j_{1}}\wedge \cdots \wedge
\varepsilon ^{j_{p}}\right\rangle \gamma (e_{j_{1}})\wedge \cdots \wedge
\gamma (e_{j_{p}}) & p\geq 2%
\end{array}%
\right.  \label{extensiondef1}
\end{equation}%
where $\left\{ e_{j}\right\} $ is any basis of $V$ and $\left\{ \varepsilon
^{j}\right\} $ is its dual basis for $V^{\star },$ and sums over the indices 
$j_{1},\ldots ,j_{p}$ from $1$ to $n$ are implicated.

We emphasize that $\left\langle x_{p},\varepsilon ^{j_{1}}\wedge \cdots
\wedge \varepsilon ^{j_{p}}\right\rangle \gamma (e_{j_{1}})\wedge \cdots
\wedge \gamma (e_{j_{p}})$ is a $p$-form which does not depend on the pair
of bases $\left( \left\{ e_{j}\right\} ,\left\{ \varepsilon ^{j}\right\}
\right) $ used for defining it, because of bilinearity of $\left\langle
\left. {}\right. ,\left. {}\right. \right\rangle $ and linearity of $\gamma
. $ Therefore, $\underline{\gamma }^{I}(x_{p})$ is a right invariant $p$%
-form which depends only on $x_{p},$ and so $\underline{\gamma }^{I}$ is
well defined. The linearity of $\underline{\gamma }^{I}$ follows from the
bilinearity of $\left\langle \left. {}\right. ,\left. {}\right.
\right\rangle .$ Thus, $\underline{\gamma }^{I}$ is a well-defined extensor
belonging to $\func{ext}\left( \bigwedge^{p}V;\bigwedge^{p}V^{\star }\right)
.$

The extension of $\gamma $ to $\func{ext}\left( \bigwedge V;\bigwedge
V^{\star }\right) $ is $\underline{\gamma }^{II}:\bigwedge V\longrightarrow
\bigwedge V^{\star }$ defined by%
\begin{equation}
\underline{\gamma }^{II}(x)\in \bigwedge V^{\star }\text{ such that }\left[ 
\underline{\gamma }^{II}(x)\right] ^{k}=\underline{\gamma }^{I}(x_{k}),\text{
for all }k=0,1,\ldots ,n  \label{extensiondef2}
\end{equation}

Notice that the linearity of $\underline{\gamma }^{II}$ follows from the
linearity of $\left[ \left. {}\right. \right] _{k}$ and the linearity of $%
\underline{\gamma }^{I}.$ Thus, $\underline{\gamma }^{II}$ is a well-defined
extensor belonging to $\func{ext}\left( \bigwedge V;\bigwedge V^{\star
}\right) .$

Because of no confusion could arise, we shall use the same notation $%
\underline{\gamma }$ for both $\underline{\gamma }^{I}$ and $\underline{%
\gamma }^{II},$ and we shall simply speak about a single extension of $%
\gamma .$

We present the properties for the extension of the metric extensor.

For all $x\in \bigwedge V:$%
\begin{equation}
\text{if }x\text{ is }p\text{-homogeneous, then }\underline{\gamma }\left(
x\right) \text{ is }p\text{-homogeneous}  \label{extension1}
\end{equation}

For all $x,y\in \bigwedge V:$%
\begin{equation}
\underline{\gamma }(x\wedge y)=\underline{\gamma }(x)\wedge \underline{%
\gamma }(y)  \label{extension2}
\end{equation}

For all $x\in \bigwedge V:$%
\begin{equation}
\widehat{\underline{\gamma }(x)}=\underline{\gamma }(\widehat{x})\text{ and }%
\widetilde{\underline{\gamma }(x)}=\underline{\gamma }(\widetilde{x})
\label{extension3}
\end{equation}%
i.e. grade involution and reversion operators commutate with extension
operator.

The extension of $\gamma $ is invertible, and%
\begin{equation}
\left\{ \underline{\gamma }\right\} ^{-1}=\underline{\left\{ \gamma
^{-1}\right\} }  \label{extension4}
\end{equation}%
i.e., the inverse of extension equals the extension of inverse. Thus, we
could use the simpler symbol $\underline{\gamma }^{-1}$ to mean both of them.

We define to follow the so-called reciprocal bases. Let $\left\{
e_{j}\right\} $ and $\left\{ \varepsilon ^{j}\right\} $ be respectively any
basis of $V$ and its dual basis for $V^{\star }.$ i.e., $\left\langle
e_{j},\varepsilon ^{k}\right\rangle =\delta _{j}^{k}.$ The sets $\left\{
e^{j}\right\} $ and $\left\{ \varepsilon _{j}\right\} $ such that $%
e^{j}=\gamma ^{-1}(\varepsilon ^{j})$ and $\varepsilon _{j}=\gamma (e_{j})$
are respectively called reciprocal bases of $\left\{ e_{j}\right\} $ and $%
\left\{ \varepsilon ^{j}\right\} ,$ because of $e_{j}\cdot e^{k}$ $%
=\varepsilon ^{k}\cdot \varepsilon _{j}=\delta _{j}^{k},$ as we prove below%
\begin{eqnarray*}
e_{j}\cdot e^{k} &=&\left\langle \gamma (e_{j}),e^{k}\right\rangle
=\left\langle e_{j},\gamma (e^{k})\right\rangle =\left\langle
e_{j},\varepsilon ^{k}\right\rangle =\delta _{j}^{k}, \\
\varepsilon ^{k}\cdot \varepsilon _{j} &=&\left\langle \gamma
^{-1}(\varepsilon ^{k}),\varepsilon _{j}\right\rangle =\left\langle
\varepsilon ^{k},\gamma ^{-1}(\varepsilon _{j})\right\rangle =\left\langle
\varepsilon ^{k},e_{j}\right\rangle =\delta _{j}^{k}.
\end{eqnarray*}

The sets $\left\{ e_{j_{1}}\wedge \cdots \wedge e_{j_{p}}\right\} _{1\leq
j_{1}<\cdots <j_{p}\leq n}$ and $\left\{ \varepsilon ^{j_{1}}\wedge \cdots
\wedge \varepsilon ^{j_{p}}\right\} _{_{1\leq j_{1}<\cdots <j_{p}\leq n}}$
are respectively bases for $\bigwedge^{p}V$ and $\bigwedge^{p}V^{\star }.$
Their reciprocal bases are respectively the sets $\left\{ e^{j_{1}}\wedge
\cdots \wedge e^{j_{p}}\right\} _{1\leq j_{1}<\cdots <j_{p}\leq n}$ and $%
\left\{ \varepsilon _{j_{1}}\wedge \cdots \wedge \varepsilon
_{j_{p}}\right\} _{_{1\leq j_{1}<\cdots <j_{p}\leq n}}.$ We disclose their
reciprocity-like properties$\footnote[1]{%
Recall the so-called generalized permutation symbol of order $p,$%
\begin{equation*}
\sigma _{j_{1}\ldots j_{p}}^{k_{1}\ldots k_{p}}=\det \left( 
\begin{array}{ccc}
\delta _{j_{1}}^{k_{1}} & \cdots & \delta _{j_{p}}^{k_{1}} \\ 
\vdots &  & \vdots \\ 
\delta _{j_{1}}^{k_{p}} & \cdots & \delta _{j_{p}}^{k_{p}}%
\end{array}%
\right) ,\text{ with }j_{1},\ldots ,j_{p}\text{ and }k_{1},\ldots ,k_{p}%
\text{ running from }1\text{ to }n
\end{equation*}%
}$ below%
\begin{eqnarray*}
\left( e_{j_{1}}\wedge \cdots e_{j_{p}}\right) \cdot \left( e^{k_{1}}\wedge
\cdots e^{k_{p}}\right) &=&\det \left( 
\begin{array}{ccc}
\delta _{_{j_{1}}}^{^{k_{1}}} & \cdots & \delta _{j_{1}}^{k_{p}} \\ 
\cdots &  & \cdots \\ 
\delta _{_{j_{p}}}^{k_{1}} & \cdots & \delta _{j_{p}}^{k_{p}}%
\end{array}%
\right) =\sigma _{j_{1}\ldots j_{p}}^{k_{1}\ldots k_{p}}, \\
\left( \varepsilon ^{k_{1}}\wedge \cdots \varepsilon ^{k_{p}}\right) \cdot
\left( \varepsilon _{j_{1}}\wedge \cdots \varepsilon _{j_{p}}\right) &=&\det
\left( 
\begin{array}{ccc}
\delta _{j_{1}}^{k_{1}} & \cdots & \delta _{j_{p}}^{k_{1}} \\ 
\cdots &  & \cdots \\ 
\delta _{j_{1}}^{k_{p}} & \cdots & \delta _{j_{p}}^{k_{p}}%
\end{array}%
\right) =\sigma _{j_{1}\ldots j_{p}}^{k_{1}\ldots k_{p}}.
\end{eqnarray*}

In particular, the unitary sets $\left\{ e_{\wedge }\right\} $ and $\left\{
\varepsilon ^{\wedge }\right\} ,$ where $e_{\wedge }=e_{1}\wedge \cdots
\wedge e_{n}$ and $\varepsilon ^{\wedge }=\varepsilon ^{1}\wedge \cdots
\wedge \varepsilon ^{n},$ are respectively bases for $\bigwedge^{n}V$ and $%
\bigwedge^{n}V^{\star }.$ Their reciprocal bases are respectively $\left\{
e^{\wedge }\right\} $ and $\left\{ \varepsilon _{\wedge }\right\} ,$ with $%
e^{\wedge }=e^{1}\wedge \cdots \wedge e^{n}=\underline{\gamma }%
^{-1}(\varepsilon ^{\wedge })$ and $\varepsilon _{\wedge }=\varepsilon
_{1}\wedge \cdots \wedge \varepsilon _{n}=\underline{\gamma }(e_{\wedge }).$

We notice to follow other two remarkable propositions for the metric
extensor.

\begin{proposition}
$\underline{\gamma }$ is a metric extensor over $\bigwedge V.$
\end{proposition}

\begin{proof}
We first prove that $\underline{\gamma }$ is symmetric. We only need to
check that $\underline{\gamma }$ satisfies the condition of symmetry for
scalars and simple $p$-vectors.

Let $\alpha ,\beta \in \mathbb{R}.$ Taking into account how $\underline{%
\gamma }$ acts on scalars, we get%
\begin{equation*}
\left\langle \underline{\gamma }(\alpha ),\beta \right\rangle =\left\langle
\alpha ,\beta \right\rangle =\left\langle \alpha ,\underline{\gamma }(\beta
)\right\rangle
\end{equation*}

Let $v_{1},\ldots ,v_{p},w_{1},\ldots ,w_{p}\in V.$ Using (\ref{extension2}%
), (\ref{Fund3}) and the symmetry of $\gamma ,$ we have%
\begin{eqnarray*}
\left\langle \underline{\gamma }(v_{1}\wedge \cdots v_{p}),w_{1}\wedge
\cdots w_{p}\right\rangle &=&\left\langle \gamma (v_{1})\wedge \cdots \gamma
(v_{p}),w_{1}\wedge \cdots w_{p}\right\rangle \\
&=&\det \left( \left\langle \gamma (v_{i}),w_{j}\right\rangle \right) =\det
\left( \left\langle v_{i},\gamma (w_{j})\right\rangle \right) \\
&=&\left\langle v_{1}\wedge \cdots v_{p},\gamma (w_{1})\wedge \cdots \gamma
(w_{p})\right\rangle \\
&=&\left\langle v_{1}\wedge \cdots v_{p},\underline{\gamma }(w_{1}\wedge
\cdots w_{p})\right\rangle .
\end{eqnarray*}

Then, the symmetry of $\underline{\gamma }$ for multivectors follows at once
from the linearity of $\underline{\gamma }$ and the bilinearity of $%
\left\langle \left. {}\right. ,\left. {}\right. \right\rangle ,$ taking into
account (\ref{scalarhomog1}), (\ref{scalarhomog2}) and (\ref{extension1}).

We now prove that $\underline{\gamma }$ is one-to-one. Let $x,y,z\in
\bigwedge V:$%
\begin{eqnarray*}
\underline{\gamma }(x) &=&\underline{\gamma }(y)\Longrightarrow \left\langle 
\underline{\gamma }(x),z\right\rangle =\left\langle \underline{\gamma }%
(y),z\right\rangle \Longrightarrow \left\langle \underline{\gamma }%
(x-y),z\right\rangle =0 \\
&\Longrightarrow &\left\langle x-y,\underline{\gamma }(z)\right\rangle =0%
\text{ for all }\underline{\gamma }(z)
\end{eqnarray*}%
then, because of non-degeneracy of $\left\langle \left. {}\right. ,\left.
{}\right. \right\rangle ,$ it follows that $x-y=0.$ Therefore, $\underline{%
\gamma }(x)=\underline{\gamma }(y)$ implies $x=y.$
\end{proof}

\begin{proposition}
$\underline{\gamma }^{-1}$ is a metric extensor over $\bigwedge V^{\star }.$
\end{proposition}

\begin{proof}
We first prove that $\underline{\gamma }^{-1}$ is symmetric. It only must be
proved that $\underline{\gamma }^{-1}$ satisfies the condition of symmetry
for scalars and simple $p$-forms.

Let $\alpha ,\beta \in \mathbb{R}.$ Taking into account how $\underline{%
\gamma }^{-1}$ acts on scalars, we get%
\begin{equation*}
\left\langle \underline{\gamma }^{-1}(\alpha ),\beta \right\rangle
=\left\langle \alpha ,\beta \right\rangle =\left\langle \alpha ,\underline{%
\gamma }^{-1}(\beta )\right\rangle
\end{equation*}

Let $\omega _{1},\ldots ,\omega _{p},\sigma _{1},\ldots ,\sigma _{p}\in
V^{\star }.$ Using the analogous equation to (\ref{extension2}) for $\gamma
^{-1},$ (\ref{Fund3}) and the symmetry of $\gamma ^{-1},$ we have%
\begin{eqnarray*}
\left\langle \underline{\gamma }^{-1}(\omega _{1}\wedge \cdots \omega
_{p}),\sigma _{1}\wedge \cdots \sigma _{p}\right\rangle &=&\left\langle
\gamma ^{-1}(\omega _{1})\wedge \cdots \gamma ^{-1}(\omega _{p}),\sigma
_{1}\wedge \cdots \sigma _{p}\right\rangle \\
&=&\det \left( \left\langle \gamma ^{-1}(\omega _{i}),\sigma
_{j}\right\rangle \right) =\det \left( \left\langle \omega _{i},\gamma
^{-1}(\sigma _{j})\right\rangle \right) \\
&=&\left\langle \omega _{1}\wedge \cdots \omega _{p},\gamma ^{-1}(\sigma
_{1})\wedge \cdots \gamma ^{-1}(\sigma _{p})\right\rangle \\
&=&\left\langle \omega _{1}\wedge \cdots \omega _{p},\underline{\gamma }%
^{-1}(\sigma _{1}\wedge \cdots \sigma _{p})\right\rangle .
\end{eqnarray*}

Then, the symmetry of $\underline{\gamma }^{-1}$ for multiforms follows
inmediately from the linearity of $\underline{\gamma }^{-1}$ and the
bilinearity of $\left\langle \left. {}\right. ,\left. {}\right.
\right\rangle ,$ considering (\ref{scalarhomog1}), (\ref{scalarhomog2}) and (%
\ref{extension1}).

We now prove that $\underline{\gamma }^{-1}$ is one-to-one. Let $\phi ,\chi
,\psi \in \bigwedge V^{\star }:$%
\begin{eqnarray*}
\underline{\gamma }^{-1}(\phi ) &=&\underline{\gamma }^{-1}(\chi
)\Longrightarrow \left\langle \underline{\gamma }^{-1}(\phi ),\psi
\right\rangle =\left\langle \underline{\gamma }^{-1}(\chi ),\psi
\right\rangle \Longrightarrow \left\langle \underline{\gamma }^{-1}(\phi
-\chi ),\psi \right\rangle =0 \\
&\Longrightarrow &\left\langle \phi -\chi ,\underline{\gamma }^{-1}(\psi
)\right\rangle =0\text{ for all }\underline{\gamma }^{-1}(\psi )
\end{eqnarray*}%
then, because of non-degeneracy of $\left\langle \left. {}\right. ,\left.
{}\right. \right\rangle ,$ it follows that $\phi -\chi =o.$ Thus, $%
\underline{\gamma }^{-1}(\phi )=\underline{\gamma }^{-1}(\chi )$ implies $%
\phi =\chi .$
\end{proof}

\section{Scalar product}

The scalar product of both multivectors and multiforms are defined by making
use of the duality pairing mapping and the extended extensors $\underline{%
\gamma }$ and $\underline{\gamma }^{-1},$ respectively.

\subsection{Scalar product of multivectors}

The scalar product of $x_{p}\in \bigwedge^{p}V$ and $y_{p}\in
\bigwedge^{p}V, $ and the scalar product of $x\in \bigwedge V$ and $y\in
\bigwedge V,$ are defined by%
\begin{equation}
x_{p}\cdot y_{p}=\left\langle \underline{\gamma }(x_{p}),y_{p}\right\rangle
\label{scalarvectors}
\end{equation}

\begin{equation}
x\cdot y=\sum_{k=0}^{n}x_{k}\cdot y_{k},\text{ i.e., }x\cdot y=\left\langle 
\underline{\gamma }(x),y\right\rangle  \label{scalarmultivectors}
\end{equation}

\subsection{Scalar product of multiforms}

The scalar product of $\phi ^{p}\in \bigwedge^{p}V^{\star }$ and $\psi
^{p}\in \bigwedge^{p}V^{\star },$ and the scalar product of $\phi \in
\bigwedge V^{\star }$ and $\psi \in \bigwedge V^{\star },$ are defined by%
\begin{equation}
\phi ^{p}\cdot \psi ^{p}=\left\langle \underline{\gamma }^{-1}(\phi
^{p}),\psi ^{p}\right\rangle  \label{scalarforms}
\end{equation}%
\begin{equation}
\phi \cdot \psi =\sum_{k=0}^{n}\phi ^{k}\cdot \psi ^{k},\text{ i.e., }\phi
\cdot \psi =\left\langle \underline{\gamma }^{-1}(\phi ),\psi \right\rangle
\label{scalarmultiforms}
\end{equation}

The scalar product has three remarkable properties: bilinearity, symmetry
and non-degeneracy. They are inmediate consequence of the respectives
properties for the duality pairing mapping, once it is taken into account
the properties of both $\underline{\gamma }$ and $\underline{\gamma }^{-1}.$

\section{Contracted products of multivectors}

The contracted products of multivectors are defined by making use of the
duality contraction mappings and the extended extensor $\underline{\gamma }.$

\subsection{Left contracted product}

The left contracted product of $y_{q}\in \bigwedge^{q}V$ by $x_{p}\in
\bigwedge^{p}V$ $(p\leq q),$ and the left contracted product of $y\in
\bigwedge V$ by $x\in \bigwedge V,$ are defined by%
\begin{equation}
x_{p}\lrcorner y_{q}=\left\langle \underline{\gamma }(x_{p}),y_{q}\right\vert
\label{lcvectors}
\end{equation}%
\begin{eqnarray}
x\lrcorner y &\in &\bigwedge V\text{ such that }\left[ x\lrcorner y\right]
_{k}=\sum_{j=0}^{n-k}x_{j}\lrcorner y_{j+k},\text{\quad }\left( k=0,1,\ldots
,n\right)  \notag \\
\text{i.e.,\qquad }x\lrcorner y &=&\left\langle \underline{\gamma }%
(x),y\right\vert  \label{lcontractedmultivectors}
\end{eqnarray}

\subsection{Right contracted product}

The right contracted product of $y_{q}\in \bigwedge^{q}V$ by $x_{p}\in
\bigwedge^{p}V$ $(p\leq q),$ and the right contracted product of $y\in
\bigwedge V$ by $x\in \bigwedge V,$ are defined by%
\begin{equation}
y_{q}\llcorner x_{p}=\left\vert y_{q},\underline{\gamma }(x_{p})\right\rangle
\label{rcvectors}
\end{equation}%
\begin{eqnarray}
y\llcorner x &\in &\bigwedge V\text{ such that }\left[ y\llcorner x\right]
_{k}=\sum_{j=0}^{n-k}y_{j+k}\llcorner x_{j},\text{\quad }\left( k=0,1,\ldots
,n\right)  \notag \\
\text{i.e.,\qquad }y\llcorner x &=&\left\vert y,\underline{\gamma }%
(x)\right\rangle  \label{rcontractedmultivectors}
\end{eqnarray}

\section{Contracted products of multiforms}

The contracted products of multiforms are defined by making use of the
duality contraction mappings and the extended extensor $\underline{\gamma }%
^{-1}.$

\subsection{Left contracted product}

The left contracted product of $\psi ^{q}\in \bigwedge^{q}V^{\star }$ by $%
\phi ^{p}\in \bigwedge^{p}V^{\star }$ $(p\leq q),$ and the left contracted
product of $\psi \in \bigwedge V^{\star }$ by $\phi \in \bigwedge V^{\star
}, $ are defined by%
\begin{equation}
\phi ^{p}\lrcorner \psi ^{q}=\left\langle \underline{\gamma }^{-1}(\phi
^{p}),\psi ^{q}\right\vert  \label{lcforms}
\end{equation}%
\begin{eqnarray}
\phi \lrcorner \psi &\in &\bigwedge V^{\star }\text{ such that }\left[ \phi
\lrcorner \psi \right] ^{k}=\sum_{j=0}^{n-k}\phi ^{j}\lrcorner \psi ^{j+k},%
\text{\quad }\left( k=0,1,\ldots ,n\right)  \notag \\
\text{i.e.,\qquad }\phi \lrcorner \psi &=&\left\langle \underline{\gamma }%
^{-1}(\phi ),\psi \right\vert  \label{lcontractedmultiforms}
\end{eqnarray}

\subsection{Right contracted product}

The right contracted product of $\psi ^{q}\in \bigwedge^{q}V^{\star }$ by $%
\phi ^{p}\in \bigwedge^{p}V^{\star }$ $(p\leq q),$ and the right contracted
product of $\psi \in \bigwedge V^{\star }$ by $\phi \in \bigwedge V^{\star
}, $ are defined by%
\begin{equation}
\psi ^{q}\llcorner \phi ^{p}=\left\vert \psi ^{q},\underline{\gamma }%
^{-1}(\phi ^{p})\right\rangle  \label{rcforms}
\end{equation}%
\begin{eqnarray}
\psi \llcorner \phi &\in &\bigwedge V^{\star }\text{ such that }\left[ \psi
\llcorner \phi \right] ^{k}=\sum_{j=0}^{n-k}\psi ^{j+k}\llcorner \phi ^{j}%
\text{\quad }\left( k=0,1,\ldots ,n\right)  \notag \\
\text{i.e.,\qquad }\psi \llcorner \phi &=&\left\vert \psi ,\underline{\gamma 
}^{-1}(\phi )\right\rangle  \label{rcontractedmultiforms}
\end{eqnarray}

The contracted products have two remarkable properties: bilinearity and
non-degeneracy. They are inmediate consequence of the respective properties
for the duality contraction mappings, once it is taken into account the
properties of both $\underline{\gamma }$ and $\underline{\gamma }^{-1}.$

\section{Identities involving the metric products}

We present remarkable identities which involve the metric products, i.e.,
the scalar product and the contracted products. The proof of some identities
can be found in Appendix.

For all $v_{1},\ldots ,v_{p},w_{1},\ldots ,w_{p}\in V,$ and $\omega
_{1},\ldots ,\omega _{p},\sigma _{1},\ldots ,\sigma _{p}\in V^{\star }:$ 
\begin{eqnarray}
\left( v_{1}\wedge \cdots \wedge v_{p}\right) \cdot \left( w_{1}\wedge
\cdots \wedge w_{p}\right) &=&\det \left( v_{i}\cdot w_{j}\right) ,
\label{extvectorscalarextvector} \\
\left( \omega _{1}\wedge \cdots \wedge \omega _{p}\right) \cdot \left(
\sigma _{1}\wedge \cdots \wedge \sigma _{p}\right) &=&\det \left( \omega
_{i}\cdot \sigma _{j}\right) .  \label{extformscalarextform}
\end{eqnarray}

For all $x_{p}\in \bigwedge^{p}V,$ $y_{q}\in \bigwedge^{q}V,$ and $\phi
^{p}\in \bigwedge^{p}V^{\star },$ $\psi ^{q}\in \bigwedge^{q}V^{\star }$ $%
\left( p\leq q\right) :$%
\begin{equation}
x_{p}\lrcorner y_{q}=\left( -1\right) ^{p(q-p)}y_{q}\llcorner x_{p},\qquad
\phi ^{p}\lrcorner \psi ^{q}=\left( -1\right) ^{p(q-p)}\psi ^{p}\llcorner
\phi ^{p}.  \label{contractionpq}
\end{equation}

For all $v,v_{1},v_{2},\ldots ,v_{p}\in V,$ and $\omega ,\omega _{1},\omega
_{2},\ldots ,\omega _{p}\in V^{\star }:$%
\begin{eqnarray}
v\lrcorner \left( v_{1}\wedge v_{2}\wedge \cdots v_{p}\right)
&=&\sum_{k-1}^{p}(-1)^{k-1}\left( v\cdot v_{k}\right) v_{1}\wedge
v_{2}\wedge \cdots \overset{\curlyvee }{v}_{k}\cdots v_{p},
\label{contrvvectors} \\
\omega \lrcorner \left( \omega _{1}\wedge \omega _{2}\wedge \cdots \omega
_{p}\right) &=&\sum_{k-1}^{p}(-1)^{k-1}\left( \omega \cdot \omega
_{k}\right) \omega _{1}\wedge \omega _{2}\wedge \cdots \overset{\curlyvee }{%
\omega }_{k}\cdots \omega _{p}.  \label{contrfforms}
\end{eqnarray}%
where $\overset{\curlyvee }{v}_{k}$ means that $v_{k}$ must be removed from $%
v_{1}\wedge v_{2}\wedge \cdots v_{p}$ but leaving the remaining vectors in
the order they had, and analogously for $\overset{\curlyvee }{\omega }_{k}.$

For all $\left( x_{p}\right) ,\left( y_{p}\right) \in \bigwedge V,$ and $%
\left( \phi ^{p}\right) ,\left( \psi ^{p}\right) \in \bigwedge V^{\star }:$%
\begin{equation}
\left( x_{p}\right) \cdot \left( y_{p}\right) =x_{p}\cdot y_{p},\qquad
\left( \phi ^{p}\right) \cdot \left( \psi ^{p}\right) =\phi ^{p}\cdot \psi
^{p}.  \label{scalarhomogmult1}
\end{equation}

For all $\left( x_{p}\right) ,\left( y_{q}\right) \in \bigwedge V$ and $%
\left( \phi ^{p}\right) ,\left( \psi ^{q}\right) \in \bigwedge V^{\star }:$%
\begin{equation}
\left( x_{p}\right) \cdot \left( y_{q}\right) =\left( \phi ^{p}\right) \cdot
\left( \psi ^{q}\right) =0,\text{ if }p\neq q  \label{scalarhomogmult2}
\end{equation}

For all $\left( x_{p}\right) ,\left( y_{q}\right) \in \bigwedge V$ and $%
\left( \phi ^{p}\right) ,\left( \psi ^{q}\right) \in \bigwedge V^{\star }:$%
\begin{equation}
\left( x_{p}\right) \lrcorner \left( y_{q}\right) ,\left( y_{q}\right)
\llcorner \left( x_{p}\right) \text{ are }\left( q-p\right) \text{%
-homogeneous multivectors, and}  \label{contrhomogmultiv1}
\end{equation}%
\begin{equation*}
\left[ \left( x_{p}\right) \lrcorner \left( y_{q}\right) \right]
_{q-p}=x_{p}\lrcorner y_{q},\text{ }\left[ \left( y_{q}\right) \llcorner
\left( x_{p}\right) \right] _{q-p}=y_{q}\llcorner x_{p},\text{ if }p\leq q
\end{equation*}%
\begin{equation}
\left( x_{p}\right) \lrcorner \left( y_{q}\right) =\left( y_{q}\right)
\llcorner \left( x_{p}\right) =0,\text{ if }p>q  \label{contrhomogmultiv2}
\end{equation}%
\begin{equation*}
\end{equation*}%
\begin{equation}
\left( \phi ^{p}\right) \lrcorner \left( \psi ^{q}\right) ,\text{ }\left(
\psi ^{q}\right) \llcorner \left( \phi ^{p}\right) \text{ are}\left(
q-p\right) \text{-homogeneous multiforms, and}  \label{contrhomogmultif1}
\end{equation}%
\begin{equation*}
\left[ \left( \phi ^{p}\right) \lrcorner \left( \psi ^{q}\right) \right]
_{q-p}=\phi ^{p}\lrcorner \psi ^{q},\text{ }\left[ \left( \psi ^{q}\right)
\llcorner \left( \phi ^{p}\right) \right] _{q-p}=\psi ^{q}\llcorner \phi
^{p},\text{ if }p\leq q
\end{equation*}%
\begin{equation}
\left( \phi ^{p}\right) \lrcorner \left( \psi ^{q}\right) =\left( \psi
^{q}\right) \llcorner \left( \phi ^{p}\right) =o,\text{ if }p>q
\label{contrhomogmultif2}
\end{equation}

For all $v\in V$ and $x,y\in \bigwedge V:$%
\begin{equation}
v\lrcorner \left( x\wedge y\right) =\left( v\lrcorner x\right) \wedge y+%
\widehat{x}\wedge \left( v\lrcorner y\right)  \label{contrvecexteriormultiv}
\end{equation}

For all $\omega \in V^{\star }$ and $\phi ,\psi \in \bigwedge V^{\star }:$%
\begin{equation}
\omega \lrcorner \left( \phi \wedge \psi \right) =\left( \omega \lrcorner
\phi \right) \wedge \psi +\widehat{\phi }\wedge \left( \omega \lrcorner \psi
\right)  \label{contrformexteriormultif}
\end{equation}

For all $x,y,z\in \bigwedge V$ and $\phi ,\chi ,\psi \in \bigwedge V^{\star
}:$%
\begin{equation}
x\lrcorner \left( y\lrcorner z\right) =\left( x\wedge y\right) \lrcorner z%
\hspace{0.5in}\left( z\llcorner y\right) \llcorner x=z\llcorner \left(
y\wedge x\right)  \label{contrcontrmultiv}
\end{equation}%
\begin{equation}
\phi \lrcorner \left( \chi \lrcorner \psi \right) =\left( \phi \wedge \chi
\right) \lrcorner \psi \hspace{0.5in}\left( \psi \llcorner \chi \right)
\llcorner \phi =\psi \llcorner \left( \chi \wedge \phi \right)
\label{contrcontrmultif}
\end{equation}%
\begin{equation}
\left( x\lrcorner y\right) \cdot z=y\cdot \left( \widetilde{x}\wedge
z\right) \hspace{0.5in}z\cdot \left( y\llcorner x\right) =\left( z\wedge 
\widetilde{x}\right) \cdot y  \label{contrscalarmultiv}
\end{equation}%
\begin{equation}
\left( \phi \lrcorner \chi \right) \cdot \psi =\chi \cdot \left( \widetilde{%
\phi }\wedge \psi \right) \hspace{0.5in}\psi \cdot \left( \chi \llcorner
\phi \right) =\left( \psi \wedge \widetilde{\phi }\right) \cdot \chi
\label{contrscalarmultif}
\end{equation}

Notice that (\ref{scalarhomogmult1}), (\ref{scalarhomogmult2}), (\ref%
{contrhomogmultiv1}), (\ref{contrhomogmultiv2}), (\ref{contrhomogmultif1})
and (\ref{contrhomogmultif2}) allow us to identify $p$-vectors with $p$%
-homogeneous multivectors, and $p$-forms with $p$-homogeneous multiforms,
whenever in doing calculations with scalar products or contracted products.

\section{Formulas involving the metric extensor}

We present remarkable formulas which involve the metric extensor and/or the
pseudoscalars $e_{\wedge },e^{\wedge }$ and $\varepsilon ^{\wedge
},\varepsilon _{\wedge }.$ The proof of some formulas can be found in
Appendix.

For all $x\in \bigwedge V$ and $\phi \in \bigwedge V^{\star }:$%
\begin{equation}
\underline{\gamma }^{-1}(\phi )\cdot x=\phi \cdot \underline{\gamma }%
(x)=\left\langle \phi ,x\right\rangle  \label{scalargamma}
\end{equation}%
\begin{equation*}
\underline{\gamma }^{-1}(\phi )\lrcorner x=\left\langle \phi ,x\right\vert
,\quad x\llcorner \underline{\gamma }^{-1}(\phi )=\left\vert x,\phi
\right\rangle
\end{equation*}%
\begin{equation}
\quad \underline{\gamma }(x)\lrcorner \phi =\left\langle x,\phi \right\vert
,\quad \phi \llcorner \underline{\gamma }(x)=\left\vert \phi ,x\right\rangle
\label{contractedgamma}
\end{equation}

For $e_{\wedge },e^{\wedge }$ and $\varepsilon ^{\wedge },\varepsilon
_{\wedge }:$%
\begin{equation}
e_{\wedge }\cdot e^{\wedge }=\varepsilon ^{\wedge }\cdot \varepsilon
_{\wedge }=1,\quad e_{\wedge }\cdot e_{\wedge }=\varepsilon _{\wedge }\cdot
\varepsilon _{\wedge },\quad \varepsilon ^{\wedge }\cdot \varepsilon
^{\wedge }=e^{\wedge }\cdot e^{\wedge }  \label{pseudoscalars1}
\end{equation}%
\begin{eqnarray}
e^{\wedge } &=&\left( \varepsilon ^{\wedge }\cdot \varepsilon ^{\wedge
}\right) e_{\wedge },\quad \varepsilon _{\wedge }=\left( e_{\wedge }\cdot
e_{\wedge }\right) \varepsilon ^{\wedge }  \notag \\
e_{\wedge } &=&\left( \varepsilon _{\wedge }\cdot \varepsilon _{\wedge
}\right) e^{\wedge },\quad \varepsilon ^{\wedge }=\left( e^{\wedge }\cdot
e^{\wedge }\right) \varepsilon _{\wedge }  \label{pseudoscalars2}
\end{eqnarray}%
\begin{equation}
\left( e_{\wedge }\cdot e_{\wedge }\right) \left( e^{\wedge }\cdot e^{\wedge
}\right) =\left( \varepsilon ^{\wedge }\cdot \varepsilon ^{\wedge }\right)
\left( \varepsilon _{\wedge }\cdot \varepsilon _{\wedge }\right) =1
\label{pseudoscalars3}
\end{equation}

For all $v,w_{1},w_{2},\ldots ,w_{p}\in V$ and $\omega ,\sigma _{1},\sigma
_{2},\ldots ,\sigma _{p}\in V^{\star }:$%
\begin{eqnarray}
\underline{\gamma }\left\langle \gamma (v),w_{1}\wedge w_{2}\wedge \cdots
\wedge w_{p}\right\vert &=&\left\langle v,\underline{\gamma }(w_{1}\wedge
w_{2}\wedge \cdots \wedge w_{p})\right\vert  \label{gamma1} \\
\underline{\gamma }^{-1}\left\langle \gamma ^{-1}(\omega ),\sigma _{1}\wedge
\sigma _{2}\wedge \cdots \wedge \sigma _{p}\right\vert &=&\left\langle
\omega ,\underline{\gamma }^{-1}(\sigma _{1}\wedge \sigma _{2}\wedge \cdots
\wedge \sigma _{p})\right\vert  \label{gamma1b}
\end{eqnarray}

For all $v\in V$ and $\omega \in V^{\star }:$%
\begin{eqnarray}
\underline{\gamma }\left( v\lrcorner e_{\wedge }\right) &=&\left( e_{\wedge
}\cdot e_{\wedge }\right) \left\langle v,\varepsilon ^{\wedge }\right\vert 
\notag \\
\left( v\lrcorner e_{\wedge }\right) \lrcorner \widetilde{e}_{\wedge }
&=&\left( e_{\wedge }\cdot e_{\wedge }\right) v  \notag \\
\gamma \left( \left\langle \omega ,e_{\wedge }\right\vert \lrcorner 
\widetilde{e}_{\wedge }\right) &=&\left( e_{\wedge }\cdot e_{\wedge }\right)
\omega  \label{gamma2}
\end{eqnarray}

For all $\omega \in V^{\star }:$%
\begin{equation}
\gamma ^{-1}(\omega )=\frac{1}{e_{\wedge }\cdot e_{\wedge }}\left\langle
\omega ,e_{\wedge }\right\vert \lrcorner \widetilde{e}_{\wedge }=\frac{1}{%
e_{\wedge }\cdot e_{\wedge }}\left\langle \omega ,\widetilde{e}_{\wedge
}\right\vert \lrcorner e_{\wedge }  \label{gamma3}
\end{equation}

For all $x,y\in \bigwedge V$ and $\phi ,\psi \in \bigwedge V^{\star }:$%
\begin{eqnarray}
\underline{\gamma }\left\langle \underline{\gamma }(x),y\right\vert
&=&\left\langle x,\underline{\gamma }(y)\right\vert  \label{gamma4} \\
\underline{\gamma }^{-1}\left\langle \underline{\gamma }^{-1}(\phi ),\psi
\right\vert &=&\left\langle \phi ,\underline{\gamma }^{-1}(\psi )\right\vert
\label{gamma4b}
\end{eqnarray}

For all $x\in \bigwedge V$ and $\phi \in \bigwedge V^{\star }:$%
\begin{eqnarray}
\underline{\gamma }\left( x\lrcorner e_{\wedge }\right) &=&\left( e_{\wedge
}\cdot e_{\wedge }\right) \left\langle x,\varepsilon ^{\wedge }\right\vert 
\notag \\
\left( x\lrcorner e_{\wedge }\right) \lrcorner \widetilde{e}_{\wedge }
&=&\left( e_{\wedge }\cdot e_{\wedge }\right) x  \notag \\
\underline{\gamma }\left( \left\langle \phi ,e_{\wedge }\right\vert
\lrcorner \widetilde{e}_{\wedge }\right) &=&\left( e_{\wedge }\cdot
e_{\wedge }\right) \phi  \label{gamma5}
\end{eqnarray}

For all $\phi \in \bigwedge V^{\star }:$%
\begin{equation}
\underline{\gamma }^{-1}(\phi )=\frac{1}{e_{\wedge }\cdot e_{\wedge }}%
\left\langle \phi ,e_{\wedge }\right\vert \lrcorner \widetilde{e}_{\wedge }=%
\frac{1}{e_{\wedge }\cdot e_{\wedge }}\left\langle \phi ,\widetilde{e}%
_{\wedge }\right\vert \lrcorner e_{\wedge }  \label{gamma6}
\end{equation}

For all $x\in \bigwedge V$ and $\phi \in \bigwedge V^{\star }:$%
\begin{eqnarray}
x &=&\frac{1}{e_{\wedge }\cdot e_{\wedge }}\left( x\lrcorner e_{\wedge
}\right) \lrcorner \widetilde{e}_{\wedge }=\frac{1}{e_{\wedge }\cdot
e_{\wedge }}\left( x\lrcorner \widetilde{e}_{\wedge }\right) \lrcorner
e_{\wedge }  \label{gamma7} \\
\phi &=&\frac{1}{\varepsilon ^{\wedge }\cdot \varepsilon ^{\wedge }}\left(
\phi \lrcorner \varepsilon ^{\wedge }\right) \lrcorner \widetilde{%
\varepsilon }^{\wedge }=\frac{1}{\varepsilon ^{\wedge }\cdot \varepsilon
^{\wedge }}\left( \phi \lrcorner \widetilde{\varepsilon }^{\wedge }\right)
\lrcorner \widetilde{\varepsilon }^{\wedge }  \label{gamma7b}
\end{eqnarray}

The equations (\ref{gamma3}) and (\ref{gamma6}) are respectively the
inversion formulas for $\gamma $ and $\underline{\gamma }.$ Notice that the
first one is a particular case of the second one. Of course, there exist
similar others involving only the right contraction and the right contracted
product.

The equations (\ref{gamma7}) and (\ref{gamma7b}) are expansion-like formulas
for multivectors and multiforms, respectively. Of course, there exist
similar others involving only the right contracted product.

\section{Conclusions}

For the so-called vector spaces of multiforms and multivectors over a finite
dimensional real vector space, we have introduced the key concept of duality
mappings. And, for any finite dimensional real vector space, we have
introduced another key concept: the metric extensor. We have disclosed the
logical equivalence between the so-called metric tensor and the metric
extensor. The duality mappings and the metric extensor allowed us to
construct the scalar product and the contracted products of both multiforms
and multivectors. The metric extensor has many advantages over the metric
tensor: (1) because of it has inverse, we can define the scalar product of
vectors and forms (2) because of properties of its extension, it is possible
to define the scalar product and the contracted products of both
multivectors and multiforms. On another side, the Clifford product of both
multivectors and multiforms is not outside, since the scalar product and the
contracted products fall within its definition (besides the exterior product
which depends only on the duality structure). So, everything is done under a
single conceptual unit: the metric extensor. This proposal unveils, once and
for all, the hidden nature of all metric product: a metric product defined
over a vector space can be visualized as a kind of deformation induced by
the metric extensor onto the duality structure of that vector space.

\section{Appendix}

We here make the proof in detail of some identities and formulas to show
which are "the tricks of the trade".

Proof of (\ref{contrvecexteriormultiv}): Using (\ref{lcontractedmultivectors}%
) and (\ref{contrformmultivectors}), we get%
\begin{eqnarray*}
v\lrcorner \left( x\wedge y\right) &=&\left\langle \gamma (v),x\wedge
y\right\vert =\left\langle \gamma (v),x\right\vert \wedge y+\widehat{x}%
\wedge \left\langle \gamma (v),y\right\vert \\
&=&\left( v\lrcorner x\right) \wedge y+\widehat{x}\wedge \left( v\lrcorner
y\right) .
\end{eqnarray*}

Proof of (\ref{contrformexteriormultif}): Using (\ref{lcontractedmultiforms}%
) and (\ref{contrvectormultiforms}), we get%
\begin{eqnarray*}
\omega \lrcorner \left( \phi \wedge \psi \right) &=&\left\langle \gamma
^{-1}(\omega ),\phi \wedge \psi \right\vert =\left\langle \gamma
^{-1}(\omega ),\phi \right\vert \wedge \psi +\widehat{\phi }\wedge
\left\langle \gamma ^{-1}(\omega ),\psi \right\vert \\
&=&\left( \omega \lrcorner \phi \right) \wedge \psi +\widehat{\phi }\wedge
\left( \omega \lrcorner \psi \right) .
\end{eqnarray*}

Proof of (\ref{contrcontrmultiv}): Using (\ref{lcontractedmultivectors}) and
(\ref{rcontractedmultivectors}), and taking into account 1st. and 2nd.
equations from (\ref{dcvectors}), (\ref{extension2}), it can be written%
\begin{eqnarray*}
x\lrcorner \left( y\lrcorner z\right) &=&\left\langle \underline{\gamma }%
(x),\left\langle \underline{\gamma }(y),z\right\vert \right\vert
=\left\langle \underline{\gamma }(x)\wedge \underline{\gamma }%
(y),z\right\vert =\left\langle \underline{\gamma }(x\wedge y),z\right\vert
=\left( x\wedge y\right) \lrcorner z, \\
\left( z\llcorner y\right) \llcorner x &=&\left\vert \left\vert z,\underline{%
\gamma }(y)\right\rangle ,\underline{\gamma }(x)\right\rangle =\left\vert z,%
\underline{\gamma }(y)\wedge \underline{\gamma }(x)\right\rangle =\left\vert
z,\underline{\gamma }(y\wedge x)\right\rangle =z\llcorner \left( y\wedge
x\right) .
\end{eqnarray*}

Proof of (\ref{contrcontrmultif}): Using (\ref{lcontractedmultiforms}) and (%
\ref{rcontractedmultiforms}), and taking into account 1st. and 2nd.
equations from (\ref{dcforms}), analogous equation to (\ref{extension2}) for 
$\underline{\gamma }^{-1},$ it can be written%
\begin{eqnarray*}
\phi \lrcorner \left( \chi \lrcorner \psi \right) &=&\left\langle \underline{%
\gamma }^{-1}(\phi ),\left\langle \underline{\gamma }^{-1}(\chi ),\psi
\right\vert \right\vert =\left\langle \underline{\gamma }^{-1}(\phi )\wedge 
\underline{\gamma }^{-1}(\chi ),\psi \right\vert =\left\langle \underline{%
\gamma }^{-1}(\phi \wedge \chi ),\psi \right\vert \\
&=&\left( \phi \wedge \chi \right) \lrcorner \psi , \\
\left( \psi \llcorner \chi \right) \llcorner \phi &=&\left\vert \left\vert
\psi ,\underline{\gamma }^{-1}(\chi )\right\rangle ,\underline{\gamma }%
^{-1}(\phi )\right\rangle =\left\vert \psi ,\underline{\gamma }^{-1}(\chi
)\wedge \underline{\gamma }^{-1}(\phi )\right\rangle =\left\vert \psi ,%
\underline{\gamma }^{-1}(\chi \wedge \phi )\right\rangle \\
&=&\psi \llcorner \left( \chi \wedge \phi \right) .
\end{eqnarray*}

Proof of (\ref{contrscalarmultiv}): Using (\ref{scalarmultivectors}), (\ref%
{lcontractedmultivectors}) and (\ref{rcontractedmultivectors}), and
considering 1st. and 2nd. equations from (\ref{dcpvectors}), (\ref%
{extension2}), (\ref{extension3}), we get%
\begin{eqnarray*}
\left( x\lrcorner y\right) \cdot z &=&\left\langle \left\langle \underline{%
\gamma }(x),y\right\vert ,\underline{\gamma }(z)\right\rangle =\left\langle
y,\underline{\gamma }(\widetilde{x})\wedge \underline{\gamma }%
(z)\right\rangle =\left\langle y,\underline{\gamma }(\widetilde{x}\wedge
z)\right\rangle =y\cdot \left( \widetilde{x}\wedge z\right) , \\
z\cdot \left( y\llcorner x\right) &=&\left\langle \underline{\gamma }%
(z),\left\vert y,\underline{\gamma }(x)\right\rangle \right\rangle
=\left\langle \underline{\gamma }(z)\wedge \underline{\gamma }(\widetilde{x}%
),y\right\rangle =\left\langle \underline{\gamma }(z\wedge \widetilde{x}%
),y\right\rangle =\left( z\wedge \widetilde{x}\right) \cdot y.
\end{eqnarray*}

Proof of (\ref{contrscalarmultif}): Using (\ref{scalarmultiforms}), (\ref%
{lcontractedmultiforms}) and (\ref{rcontractedmultiforms}), and considering
1st. and 2nd. equations from (\ref{dcpforms}), analogous equation to (\ref%
{extension2}) for $\underline{\gamma }^{-1},$ analogous equation to (\ref%
{extension3}) for $\underline{\gamma }^{-1},$ we get%
\begin{eqnarray*}
\left( \phi \lrcorner \chi \right) \cdot \psi &=&\left\langle \left\langle 
\underline{\gamma }^{-1}(\phi ),\chi \right\vert ,\underline{\gamma }%
^{-1}(\psi )\right\rangle =\left\langle \chi ,\underline{\gamma }^{-1}(%
\widetilde{\phi })\wedge \underline{\gamma }^{-1}(\psi )\right\rangle \\
&=&\left\langle \chi ,\underline{\gamma }^{-1}(\widetilde{\phi }\wedge \psi
)\right\rangle =\chi \cdot \left( \widetilde{\phi }\wedge \psi \right) , \\
\psi \cdot \left( \chi \llcorner \phi \right) &=&\left\langle \underline{%
\gamma }^{-1}(\psi ),\left\vert \chi ,\underline{\gamma }^{-1}(\phi
)\right\rangle \right\rangle =\left\langle \underline{\gamma }^{-1}(\psi
)\wedge \underline{\gamma }^{-1}(\widetilde{\phi }),\chi \right\rangle \\
&=&\left\langle \underline{\gamma }^{-1}(\psi \wedge \widetilde{\phi }),\chi
\right\rangle =\left( \psi \wedge \widetilde{\phi }\right) \cdot \chi .
\end{eqnarray*}

Proof of (\ref{gamma1}): Using (\ref{contrformvectors}), symmetry of $\gamma
,$ linearity of $\underline{\gamma },$ (\ref{extension2}) and (\ref%
{contrvectorforms}), we have%
\begin{eqnarray*}
\underline{\gamma }\left\langle \gamma (v),w_{1}\wedge \cdots
w_{p}\right\vert &=&\underline{\gamma }\left(
\sum_{k=1}^{p}(-1)^{k-1}\left\langle \gamma (v),v_{k}\right\rangle
w_{1}\wedge \cdots \overset{\curlyvee }{w}_{k}\cdots w_{p}\right) \\
&=&\sum_{k=1}^{p}(-1)^{k-1}\left\langle v,\gamma (w_{k})\right\rangle \gamma
(w_{1})\wedge \cdots \gamma (\overset{\curlyvee }{w}_{k})\cdots \gamma
(w_{p}) \\
&=&\left\langle v,\gamma (w_{2})\wedge \cdots \gamma (w_{p})\right\vert
=\left\langle v,\underline{\gamma }(w_{1}\wedge \cdots w_{p})\right\vert .
\end{eqnarray*}

Proof of (\ref{gamma2}): Using (\ref{gamma1}), 2nd. equation from (\ref%
{pseudoscalars2}) and linearity of $\left\langle \left. {}\right. ,\left.
{}\right. \right\vert ,$ we have%
\begin{equation}
\underline{\gamma }\left( v\lrcorner e_{\wedge }\right) =\left\langle v,%
\underline{\gamma }(e_{\wedge })\right\vert =\left\langle v,\varepsilon
_{\wedge }\right\vert =\left\langle v,\left( e_{\wedge }\cdot e_{\wedge
}\right) \varepsilon ^{\wedge }\right\vert =\left( e_{\wedge }\cdot
e_{\wedge }\right) \left\langle v,\varepsilon ^{\wedge }\right\vert . 
\tag{i}
\end{equation}

Using (\ref{lcvectors}), (i), linearity of $\left\langle \left. {}\right.
,\left. {}\right. \right\vert $ and 2nd. equation from (\ref%
{expansionformula1}), we can write 
\begin{eqnarray}
\left( v\lrcorner e_{\wedge }\right) \lrcorner \widetilde{e}_{\wedge }
&=&\left\langle \underline{\gamma }\left( v\lrcorner e_{\wedge }\right) ,%
\widetilde{e}_{\wedge }\right\vert =\left\langle \left( e_{\wedge }\cdot
e_{\wedge }\right) \left\langle v,\varepsilon ^{\wedge }\right\vert ,%
\widetilde{e}_{\wedge }\right\vert =\left( e_{\wedge }\cdot e_{\wedge
}\right) \left\langle \left\langle v,\varepsilon ^{\wedge }\right\vert ,%
\widetilde{e}_{\wedge }\right\vert  \notag \\
&=&\left( e_{\wedge }\cdot e_{\wedge }\right) v.  \TCItag{ii}
\end{eqnarray}

Let $v=\gamma ^{-1}(\omega )$ in (ii), it yields%
\begin{equation}
\left( \gamma ^{-1}(\omega )\lrcorner e_{\wedge }\right) \lrcorner 
\widetilde{e}_{\wedge }=\left( e_{\wedge }\cdot e_{\wedge }\right) \gamma
^{-1}(\omega )  \tag{iii}
\end{equation}%
but, using (\ref{lcvectors}), we have $\gamma ^{-1}(\omega )\lrcorner
e_{\wedge }=\left\langle \gamma \gamma ^{-1}(\omega ),e_{\wedge }\right\vert
=\left\langle \omega ,\right\vert e_{\wedge }.$ Then, putting this result
into (iii), we get%
\begin{equation*}
\left\langle \omega ,e_{\wedge }\right\vert \lrcorner \widetilde{e}_{\wedge
}=\left( e_{\wedge }\cdot e_{\wedge }\right) \gamma ^{-1}(\omega ),\text{%
\quad i.e., }\gamma \left( \left\langle \omega ,e_{\wedge }\right\vert
\lrcorner \widetilde{e}_{\wedge }\right) =\left( e_{\wedge }\cdot e_{\wedge
}\right) \omega .
\end{equation*}

Proof of (\ref{gamma3}): We see that $e_{\wedge }\cdot e_{\wedge }\neq 0,$
because of 2nd. equation from (\ref{pseudoscalars3}), and $\left\langle
\omega ,e_{\wedge }\right\vert \lrcorner \widetilde{e}_{\wedge }$ is
certainly a vector. We check that $\gamma ^{-1}\gamma (v)=v$ and $\gamma
\gamma ^{-1}(\omega )=\omega .$

Using (\ref{lcvectors}) and 1st. equation from (\ref{gamma2}), we get%
\begin{equation*}
\gamma ^{-1}\gamma (v)=\gamma ^{-1}\left[ \gamma (v)\right] =\frac{1}{%
e_{\wedge }\cdot e_{\wedge }}\left\langle \gamma (v),e_{\wedge }\right\vert
\lrcorner \widetilde{e}_{\wedge }=\frac{1}{e_{\wedge }\cdot e_{\wedge }}%
\left( v\lrcorner e_{\wedge }\right) \lrcorner \widetilde{e}_{\wedge }=v,
\end{equation*}%
and by linearity of $\gamma $ and use of 3rd. equation from (\ref{gamma2}),
we get 
\begin{equation*}
\gamma \gamma ^{-1}(\omega )=\gamma \left[ \gamma ^{-1}(\omega )\right]
=\gamma \left( \frac{1}{e_{\wedge }\cdot e_{\wedge }}\left\langle \omega
,e_{\wedge }\right\vert \lrcorner \widetilde{e}_{\wedge }\right) =\frac{1}{%
e_{\wedge }\cdot e_{\wedge }}\gamma \left( \left\langle \omega ,e_{\wedge
}\right\vert \lrcorner \widetilde{e}_{\wedge }\right) =\omega .
\end{equation*}

{\large Acknowledgements: }I am very grateful to Dr. W. A. Rodrigues, Jr.
for his enlightening discussions in the past and his interest in my research
work today. I would like to express my most sincere gratitude to Lic. M. T.
Bernasconi, certainly, this paper could never have been written without her
loving spiritual support.

\end{document}